\documentclass[a4paper,twoside,12pt]{article}
\usepackage[english]{babel}
\usepackage{bbding}
\usepackage[latin1]{inputenc}
\usepackage[T1]{fontenc}
\usepackage{amsmath,amssymb,color,epsfig}
\usepackage{mathrsfs}
\usepackage{eufrak}
\usepackage{dsfont}
\usepackage{upgreek}
\usepackage{fancyhdr}
\usepackage{amsthm}
\newtheorem{theorem}{Theorem}[section]
\newtheorem{proposition}[theorem]{Proposition}
\newtheorem{corollary}[theorem]{Corollary}
\newtheorem{remark}[theorem]{Remark}
\newtheorem{lemma}[theorem]{Lemma}
\newtheorem{definition}[theorem]{Definition}
\newtheorem{Example}[theorem]{Example}

\def\1{\mathds{1}}

\newcommand{\footremember}[2]{%
   \footnote{#2}
    \newcounter{#1}
    \setcounter{#1}{\value{footnote}}%
}

\title{Statistical density of particles in one dimensional interaction and Jellium Model}
\author{%
    Mohamed Bouali\footremember{alley}{Department of mathematics, Preparatory Institute for Engineering Studies of Tunis\newline
     Department of mathematics, Faculty of Sciences of Tunis\newline E-mail: bouali25@laposte.net}%
}

\date{}
\begin{document}


\maketitle
\begin{abstract}
We study a one-dimensional gas of \( n \) charged particles confined by a potential and interacting through the Riesz potential or a more general potential. In equilibrium, and for symmetric potential the particles arrange themselves symmetrically around the origin within a finite region. Various models will be studied by modifying both the confining potential and the interaction potential. Focusing on the statistical properties of the system, we analyze the position of the rightmost particle, \( x_{\text{max}} \), and show that its typical fluctuations are described by a limiting distribution different from the Tracy-Widom distribution found in the one-dimensional log-gas. We also derive the large deviation functions governing the atypical fluctuations of \( x_{\text{max}} \) far from its mean.
\end{abstract}
\section{Introduction}
The Riesz gas is one of the most general and widely applicable models, ideal for the
study of particles governed by long-range interactions \eqref{1} - for recent reviews with
mathematical perspectives see \cite{2, 3}. It is a system composed of $n$ particles that
interact via pairwise interactions and are confined by an external potential.
The pairwise repulsive interactions vary with the distance as a power law and therefore
the energy of the gas is given by
\begin{equation}\label{1}
E(x_1,...,x_n)=\sum_{k=1}^nQ(x_i)- {\rm sign}(k)\sum_{1\leq i\neq j\leq n}|x_i-x_j|^k,\quad k<2.
\end{equation}
Various special integer values of $k$ have been studied before. For example, by
setting $k\to 0^+$, which results in a pairwise repulsion that varies as the logarithm of
the distance, the Dyson's log-gas is recovered \cite{4, 5, 6}. Its special property is that the
positions of the particles can be mapped to the eigenvalues of an invariant random
matrices model. The particular potential $Q(x)=x^2$ corresponds to the Gaussian random matrix ensemble. The connection is made by identifying the Boltzman weight of the gas with the joint distribution of eigenvalues of an $n\times n$ invariant random matrix  ensembles. Another known model that is part of the Riesz gas family is the classical Calogero-Moser model. Its energy is given by Eq. (1)
with $k = 2$ and it is exactly solvable \cite{7, 8, 9}. Also, the case \(k=1\) which corresponds to one dimensional one
component plasma also known as the jellium model was studied for the the particular potential $Q(x)=x^2$ in \cite{12, 13, 14, 15}. For this model
most of the earlier studies considered bulk properties at the thermodynamic limit. In \cite{16}, the authors addressed the extreme value question in the 1d OCP or the
"jellium" model where they showed analytically that the limiting distribution of the
typical fluctuations of $x_{\text max}$ is indeed different from the Tracy-Widom distribution.

Let \( \Sigma \) be a finite or infinite interval of \( \mathbb{R} \), and consider the more general probability measure on \( \Sigma^n \):
\[
\mathbb{P}_n(dx) = \frac{1}{Z_{\Sigma, n}} e^{-E(x_1, \dots, x_n)} dx_1 \dots dx_n,
\]
where \[E(x_1,...,x_n)=\sum_{k=1}^nQ(x_i)- \sum_{1\leq i\neq j\leq n}V(x_i-x_j),\]
and $Q$, $V$ are two functions to be specified later.

The empirical density of the particles \( (x_1, \dots, x_n) \in\Sigma^n\) is defined by
\[
\nu_{\Sigma, n} = \frac{1}{n} \sum_{i=1}^n \delta_{x_i},
\]
and the statistical density is given by \( \mu_{\Sigma, n} = \mathbb{E}(\nu_{\Sigma, n}) \). That is, for a continuous function \( f \) on \( \Sigma \),
\[
\int_\Sigma f(t) \mu_{\Sigma, n}(dt) = \frac{1}{n} \int_{\Sigma^n} \sum_{i=1}^n f(x_i) \mathbb{P}_n(dx).
\]
One is interested to the statistical distribution as $n\to\infty$ for a sufficiently large class of potentials \( Q \) and interactions \( V \).
Which leads to the
following problem of potential theory. Consider the energy
\[
E(\mu) = \int_\Sigma Q(x) \mu(dx) - \int_{\Sigma \times \Sigma} V(x - y) \mu(dx) \mu(dy).
\]
We want to find the probability measure $\mu$ which realizes the minimum of the
energy $E(\mu)$.

Under the conditions that \( V \) is a negative definite and lower semicontinuous function satisfying
\[
\int_{\mathbb{R}^2} V(x - y) \mu(dx) \mu(dy) < 0
\]
for every signed measure \( \mu \) with compact support and \( \mu(\mathbb{R}) = 0 \), and \( Q \) is a lower semicontinuous function with
\[
\lim_{|x| \to \infty} Q(x) - 4V(x) = +\infty,
\]
we show that there exists a unique probability measure \( \mu^* \) with compact support such that
\[
\inf_{\mu \in \mathcal{M}^1(\Sigma)} E(\mu) = E(\mu^*),
\]
where \( \mathcal{M}^1(\Sigma) \) is the space of probability measures on \( \Sigma \). Furthermore, we prove that, after scaling the particles, the probability measure \( \mu_n \) converges in the tight topology to the equilibrium measure \( \mu^* \).

Next, we explicitly compute the density of the equilibrium measure \( \mu^* \) for various potentials \( Q \) and kernels \( V_k \), \( k < 2 \), where
\[
V_k(x) =
\begin{cases}
\text{sign}(k) |x|^k, & k \neq 0, \\
\log |x|, & k = 0.
\end{cases}
\]
Our first result addresses the case \( k = 1 \) and \( Q \) a twice differentiable and convex function with
\[
\lim_{|x| \to \infty} Q(x) - 4|x| = +\infty.
\]
We show that the density of \( \mu^* \) with respect to the Lebesgue measure is given by
\[
g(x) = \frac{1}{2} Q''(x) \chi_{[a, b]}(x),
\]
where \( \chi_{[a, b]} \) is the characteristic function of the interval \( [a, b] \). The parameters \( a \) and \( b \) are determined by the conditions
\[
Q'(b) - Q'(a) = 2, \quad Q(b) - Q(a) = a + b.
\]

The second result concerns the case \( k \in (-1, 2) \), \( k\neq 0 \) and \( Q(x) = |x|^\beta \), \( \beta > 1 \), corresponding to the generalized Ullman density. Let
\[
E_k(\mu) = \int_\mathbb{R} |x|^\beta \mu(dx) - \int_\mathbb{R} V_k(x - y) \mu(dx) \mu(dy),
\]
and let \( \mu^*_{k, \beta} \) be the measure that minimizes this energy.

For \( k \in (-1, 1) \), \(k\neq 0\) and \( \beta > 1 \), the density of \( \mu^*_{k, \beta} \) with respect to the Lebesgue measure is supported on \( [-A_{k, \beta}, A_{k, \beta}] \) and is given by
\[
f_{k, \beta}(x) = \frac{1}{C_{k, \beta}} |x|^{\beta - k - 1} \int_1^{A_{k, \beta} / |x|} \frac{u^{\beta - 1}}{(u^2 - 1)^{(1 + k)/2}} du,
\]
where
\[
C_{k, \beta} = \sqrt{\pi} \frac{|k| \Gamma\left(\frac{1 - k}{2}\right) \Gamma\left(\frac{1 + k}{2}\right)}{\Gamma\left(1 - \frac{k}{2}\right)} \frac{\Gamma\left(\frac{\beta - k}{2}\right)}{\beta \Gamma\left(\frac{\beta + 1}{2}\right)},
\]
and
\[
A_{k, \beta} = \left( \frac{2 |k| \Gamma\left(\frac{1 + k}{2}\right) \Gamma\left(1 + \frac{\beta - k}{2}\right)}{\beta \Gamma\left(\frac{\beta + 1}{2}\right)} \right)^{\frac{1}{\beta - k}}.
\]
For \( k \in (1, 2) \) and \( \beta > k \), the density is
\[
g_{k, \beta}(x) = \frac{1}{C'_{k, \beta}} |x|^{\beta - k - 1} \int_1^{A'_{k, \beta} / |x|} \frac{u^{\beta - 2}}{(u^2 - 1)^{k/2}} du,
\]
with \( \mu^*_{k, \beta} \) supported on \( [-A'_{k, \beta}, A'_{k, \beta}] \) and density \( g_{k, \beta} \), where
\[
C'_{k, \beta} = \sqrt{\pi} \frac{k (k - 1) \Gamma\left(\frac{k}{2}\right) \Gamma\left(1 - \frac{k}{2}\right)}{\beta (\beta - 1) \Gamma\left(\frac{3 - k}{2}\right)} \frac{\Gamma\left(\frac{\beta - k}{2}\right)}{\Gamma\left(\frac{\beta}{2}\right)},
\]
and
\[
A'_{k, \beta} = \left( \frac{k (k - 1) \Gamma\left(\frac{k}{2}\right) \Gamma\left(1 + \frac{\beta - k}{2}\right)}{(\beta - 1) \Gamma\left(1 + \frac{\beta}{2}\right)} \right)^{\frac{1}{\beta - k}}.
\]
For \( \beta > 1 \) and \( k = 1 \), the density of \( \mu^*_{1, \beta} \) is given by
\[
h_\beta(x) = \frac{\beta (\beta - 1)}{4} |x|^{\beta - 2} \chi_{[-(\frac{2}{\beta})^{1/(\beta - 1)}, (\frac{2}{\beta})^{1/(\beta - 1)}]}(x).
\]

The particular case \( \beta = 2 \) and \( k \in (-1, 2) \), \(k\neq 0\) generalizes the Wigner semicircle law, with the density of the equilibrium measure given by
\[
f_k(x) = \frac{1}{|k| C_k} (A_k^2 - x^2)^{\frac{1 - k}{2}} \chi_{[-A_k, A_k]}(x),
\]
where
\[
C_k = \frac{\pi (1 - k)}{\sin\left(\frac{\pi (1 - k)}{2}\right)},
\]
and
\[
A_k = \left( \frac{|k| C_k \Gamma\left(\frac{4 - k}{2}\right)}{\sqrt{\pi} \Gamma\left(\frac{3 - k}{2}\right)} \right)^{\frac{1}{2 - k}}.
\]
For \( k = 0 \) and \(\beta>1\), we retrieve the Ullman's density.
For \( k = 0 \) and \(\beta=2\), we recover Wigner's semicircle law:
\[
f_0(x) = \frac{1}{\pi} \sqrt{2 - x^2} \chi_{[-\sqrt{2}, \sqrt{2}]}(x),
\]

Our third result addresses the case \( \Sigma_\omega = [\omega, +\infty) \) and \( Q(x) = x^2 \), providing an explicit computation of the density of the equilibrium measure \( \mu^*_{k, \omega} \) that minimizes the energy
\[
E_{k, \omega}(\mu) = \int_\omega^\infty x^2 \mu(dx) - \int_\omega^\infty \int_\omega^\infty V_k(x - y) \mu(dx) \mu(dy).
\]
For \( k \in (-1, 1) \), \( k \neq 0 \), the density is
\[
g_{k, \omega}(t) = \frac{1}{|k| C_k} (L_k(\omega) + \omega - t)^{\frac{1 - k}{2}} (t - \omega)^{-\frac{1 + k}{2}} \left( t - k \omega + \frac{1 - k}{2} L_k(\omega) \right).
\]
The probability measure \( \mu_{k, \omega} \) admits the density \( g_{k, \omega} \) and is supported on \( [\omega, \omega + L_k(\omega)] \), where
\[
C_k = \frac{\pi (1 - k)}{\sin\left(\frac{(1 - k) \pi}{2}\right)},
\]
and \( L_k(\omega) \) is the unique positive solution to the equation in \( x \):
\[
(3 - k) x^{2 - k} + 2 (2 - k) \omega x^{1 - k} = \frac{2^{2 - k} |k| C_k \Gamma\left(\frac{4 - k}{2}\right)}{\sqrt{\pi} \Gamma\left(\frac{3 - k}{2}\right)}.
\]
The above result generalizes the Dean-Majumdar result presented in \cite{dean}, which specifically considers the case where \(k=0\). In that case the density turns to
\[g_{0, \omega}(t) = \frac{1}{2\pi} \sqrt{\frac{L_0(\omega) + \omega - t} {t - \omega}} \left( 2t + L_0(\omega) \right).
\]
The probability measure \( \mu^*_{0, \omega} \) admits the density \( g_{0, \omega} \) and is supported on \( [\omega, \omega + L_0(\omega)] \), where
\[ L_0(\omega)= \frac23(\sqrt{\omega^2+6}-\omega).\]
The last result that we consider is for $Q(x)=x$ and $\Sigma=[0,+\infty)$. We show that the equilibrium measure $\mu_k^*$ which minimizes the energy
\[
E_{k}(\mu) = \int_{\Bbb R_+} x\, \mu(dx) - \int_{\Bbb R^2_+}V_k(x - y) \mu(dx) \mu(dy),
\]
admits the following density:
For \( k \in (-1, 1) \) with \( k \neq 0 \),
\[
h_k(t) = \frac{\sin\left(\frac{(1-k)\pi}{2}\right)}{2|k|\pi} (a_k - t)^{\frac{1-k}{2}} t^{-\frac{1+k}{2}} \chi_{[0, a_k]}(t),
\]
where
\[
a_k = \left(\frac{|k| \Gamma\left(\frac{1+k}{2}\right) \Gamma\left(\frac{2-k}{2}\right)}{2^{k-2} \sqrt{\pi}}\right)^{\frac{1}{1-k}}.
\]
For \( k = 0 \), the density is:
\[
h_0(t) = \frac{1}{2\pi} \sqrt{\frac{4 - t}{t}} \chi_{[0, 4]}(t).
\]
This result can, in some sense, be interpreted as a generalization of the Marchenko-Pastur distribution.
 \section{Generalized repulsion}

 A function \( V: \mathbb{R} \to \mathbb{R} \) is said to be  {\it negative definite} if it is even, i.e., \( V(-x) = V(x) \), for any \( n \in \mathbb{N} \), any real numbers \( x_1, \dots, x_n \), and any complex numbers \( c_1, \dots, c_n \) satisfying \( \sum_{i=1}^n c_i = 0 \), the following inequality holds:
\[
\sum_{i,j=1}^n c_i \overline{c_j} V(x_i - x_j) \leq 0.
\]

Let \( Q: \mathbb{R} \to \mathbb{R} \) be a lower semicontinous function, \( \Sigma \subseteq \mathbb{R} \) be a finite or infinite closed interval, and \( \mathcal{M}^1(\Sigma) \) the set of probability measures on \( \Sigma \). For a probability measure \( \mu \in \mathcal{M}^1(\Sigma) \), the \( V \)-potential \( U^\mu \) is defined as:
\[
U^\mu(x) = -\int_\Sigma V(x - y) \, \mu(dy).
\]
The energy \( E(\mu) \) of the measure \( \mu \) is given by:
\[
E(\mu) = \int_\Sigma Q(x) \, \mu(dx) + \int_\Sigma U^\mu(x) \, \mu(dx).
\]
This can be rewritten as:
\[
E(\mu) = \int_\Sigma Q(x) \, \mu(dx) - \int_{\Sigma^2} V(x - y) \, \mu(dx) \, \mu(dy).
\]
Equivalently, the energy can be expressed in terms of the kernel \( k(x, y) \):
\[
E(\mu) = \int_{\Sigma^2} k(x, y) \, \mu(dx) \, \mu(dy),
\]
where the kernel \( k(x, y) \) is defined as:
\[
k(x, y) = \frac{1}{2} Q(x) + \frac{1}{2} Q(y) - V(x - y).
\]


Let \(\Sigma \) be a closed interval \(\Sigma = \Bbb R,\, [a,b[,\, ] a, b]\, {\rm or}\, [a, b]\), and \(Q\) a
function defined on \(\Sigma\) with values on \(] -\infty,\infty]\). If
$\Sigma$ is unbounded, it is assumed that
\[\displaystyle\lim_{|x|\to\infty} Q(x)-4V(x)=+\infty.\]
\begin{lemma} Assume that $Q$ and $V$ are lower-semicontinuous on $\Sigma$. Then, the energy $E$ is bounded from below.
The equilibrium energy is defined by
$$E^*_\Sigma=\inf \{E(\mu )\mid \mu \in \mathcal{M}^1(\Sigma)\},$$
\end{lemma}

\begin{definition} A function $V$ will be called {\it admissible} if it is lower semicontinuous, negative definite and for every singed measure with compact support and $\mu(\Sigma)=0$, we have
$$\int_{\Sigma^2}V(x-y)\mu(dx)\mu(dy)<0.$$
\end{definition}
\begin{corollary}\label{lem2}
Every non zero continuous negative definite function is admissible.
\end{corollary}
\begin{proof}
The function $V$ is continuous negative definite and real. By the Levy-Khinchin formula, there exists a constant $b\geq 0$ and positive measure $\nu$ on $\Bbb R^*$ such that
$$V(x)=V(0)+b x^2+\int_{\mathbb R^*}(1-\cos(xt))\nu(dt),$$
with $\displaystyle\int_{\mathbb R^*}\frac{t^2}{1+t^2}\nu(dt)<\infty$.
Therefore,
$$\int _{\Sigma^2}V( x-y)\mu (dx)\mu (dy)=-2b\Bigl(\int_{\Sigma}x\mu(dx)\Bigr)^2-\int_{\mathbb R^*}\int _{\Sigma^2}\cos((x-y)t)\mu (dx)\mu (dy)\nu(dt).$$
Since,
$$\int _{\Sigma^2}\cos((x-y)t)\mu (dx)\mu (dy)=|\hat\mu_\Sigma(t)|^2.$$
where $\hat \mu_\Sigma $ is the Fourier transform of $\mu\chi_\Sigma $, and $\chi_\Sigma$ the characteristic function of the set \(\Sigma\):
$$\hat \mu_\Sigma (t)=\int_{\Sigma} e^{itx}\mu (dx).$$
Then,
$$\int _{\Sigma^2}V( x-y)\mu (dx)\mu (dy)=-2b\Bigl(\int_{\Sigma}x\mu(dx)\Bigr)^2-\int_{\mathbb R^*}|\hat\mu_\Sigma(t)|^2\nu(dt)< 0.$$
\end{proof}
\begin{theorem}\label{thm3}  Assume that $Q$ is continuous and $V$ is an admissible function, and $\displaystyle\lim_{|x|\to\infty} Q(x)-4V(x)=+\infty$. Then,
there exists a unique measure $\mu \in \mathcal{M}^1(\Sigma)$ whose energy reaches the infimum:
$$E(\mu)=E^*_\Sigma.$$
It is called the 
equilibrium measure, and denoted by
$\mu^*_\Sigma$. Moreover the support of $\mu^*_\Sigma$ is compact.

\end{theorem}

We recall first three properties of the energy.

(P1)  Define $g_V(x)=Q(x)-4V(x)$. For $\mu \in \mathcal{M}^1(\Sigma)$,
$$\int _{\Sigma} g_V(x) \mu (dx)\leq E(\mu ).$$

(P2)  For a sequence $(\mu _n)$ in $\mathcal{M}^1(\Sigma)$ converging to a measure $\mu $
for the tight topology,
$$E(\mu )\leq \liminf_{n\to \infty } E(\mu _n).$$
This means that the map
$$\mu \mapsto E(\mu ),\quad \mathcal{M}^1(\Sigma) \to \mathbb{R}.$$
is lower semicontinuous.

(P3)  As a consequence of Prokhorov's criterium, for $C> E^*_\Sigma$, the set
$$M_C=\{\mu \in \mathcal{M}^1(\Sigma) \mid E(\mu )\leq C\}.$$
is compact for the tight topology.

\bigskip

\noindent
\begin{proof}[Proof of Theorem~{\upshape\ref{thm3}}]

\medskip

We follow the method of proof of Theorem 6.27 in \cite{10}. (See also Theorem I.2 in \cite{11}).

\medskip

a) {\it Existence of the equilibrium measure}

\medskip

We first prove that, for $C> E^*_\Sigma$, the set
$$M_{C}=\{\mu \in \mathcal{M}^1(\Sigma) \mid E(\mu )\leq C\}$$
is compact. By property (P2) $M_C$ is closed.
Since, $V$ is a negative definite function. Then, $|V(x-y)|\leq |V(x)|+|V(y)|+2\sqrt{|V(x)}\sqrt{|V(y)|}\leq 2(|V(x)|+|V(y)|) $. Then,
\begin{equation}\label{k}k(x,y)\geq {1\over 2}(g_V(x)+g_V(y)).\end{equation}
Therefore, for $\mu\in M_C$,
$$\int_\Sigma g_V(x)\leq C.$$
Since, $\lim_{|x|\to\infty }g_V(x)=+\infty$. By Prokhorov Criterium the set $M_C$ is compact.


The map $\mu \mapsto E(\mu )$ is lower semicontinuous, therefore reaches its infimum
on the compact set $M_{C}$.

\bigskip

b) {\it The support of a measure $\mu \in \mathcal{M}^1(\Sigma)$ with
$E(\mu )=E^*_\Sigma$ is compact.}

\medskip

For a measurable set $B\subset \Sigma$, let $\chi$ be its characteristic function.
For $t>-1$, define
$$\mu _t=\frac{1+t\chi }{1+t\mu (A)}\mu.$$
The measure $\mu _t$ belongs to $\mathcal{M}^1(\Sigma)$, and equals $\mu $
for $t=0$. The function $t\mapsto E(\mu _t)$ reaches its minimum for $t=0$. Hence
$$\frac{d}{dt}E(\mu _t)\big\vert_{t=0}=0.$$
Let us compute $E(\mu _t)$:
\begin{eqnarray*}
E(\mu _t)
&=&\int _{\Sigma^2} k(x,y)\frac{(1+t\delta \chi (x))(1+t\delta \chi (y))}
{(1+t\mu (A))^2}\mu (dx)\mu  (dy),
\end{eqnarray*}
and its derivative at $t=0$:
\begin{eqnarray*}
\frac{d}{dt} E(\mu _t)\big\vert_{t=0}
&=&\int _{\Sigma^2} k(x,y)\bigl(\chi (x)+\chi (y)\bigr)\mu(dx)\mu(dy)
-2\mu (A)E(\mu _1)\int_{\Sigma^2} k(x,y)\mu(dx)\mu(dy).
\end{eqnarray*}
By using the inequality \eqref{k},
we obtain
$$\int _A \Bigl(g_V(x)+ \int _{\Sigma} g_V(y)\mu(dy) -2E(\mu)\Bigr)
\mu(dx)\leq 0.$$
Since $g_V(x)\to \infty $ as $x\to \infty $, choosing $a$ sufficiently large, then for $|x|\geq a$.
$$g_V(x)+ \int _{\Sigma} g_V(y)\mu(dy) -2E(\mu)>0.$$
Let $A=\Sigma\backslash[-a,a]$, then $\mu(A)=0$ and the support of $\mu$ is a subset of $[-a,a]$.

\bigskip

c) {\it Unicity of the equilibrium measure}


\medskip

 Consider two probability measures $\mu _0$ and $\mu _1$ with compact support, and, for $t\in \mathbb{R}$,
$$\mu _t=(1-t)\mu _0+t\mu _1.$$
Then
$${d^2\over dt^2}E(\mu _t)=-2\int _{\mathbb{R}^2} V( x-y)\nu (dx)\nu (dy),$$
with $\nu =\mu _1-\mu _0$.
It follows from the hypothesis 
that the map $\mu \mapsto E(\mu )$, defined on the set
of probability measures with compact support,
is strictly convex. This implies uniqueness.
\end{proof}
\section{Determination of the equilibrium measure}\label{sec5}

Observing that the equilibrium measure $\mu^*_\Sigma$ is essentially a critical point for the energy, we will prove:

\begin{proposition}\label{p3} We keep the same hypothesis of Theorem \ref{thm3}.
Let $\mu $ be a probability measure. Assume that its potential $U^{\mu }$ is continuous and that there exists a constants $C$ such that
\begin{eqnarray*}
(1)\quad &U^{\mu } (x)+{1\over 2} Q(x)&=C\ {\rm on}\  {\rm supp}(\mu ),\cr
(2)\quad &U^{\mu } (x)+{1\over 2} Q(x)&\geq C\  {\rm on}\ \Sigma\backslash {\rm supp}(\mu ).
\end{eqnarray*}
Then $\mu $ is the equilibrium measure $\mu^*_\Sigma$.
\end{proposition}

\begin{proof}
By using the formula
$$E(\mu +\nu )
=E(\mu )+2\int _{\Sigma}\bigl(U^{\mu } (x)+{1\over 2} Q(x)\bigr)\nu (dx)$$
with $\nu =\mu^* -\mu $, we obtain
$$E(\mu^*_\Sigma)=E(\mu )+2\int _{\Sigma} \bigl(U^{\mu } (x)+{1\over 2}Q(x)\bigr)
(\mu^*_\Sigma-\mu )(dx)
-\int _{\mathbb{R}^2} V(x-y)
(\mu^*_\Sigma-\mu )(dx)(\mu^*_\Sigma-\mu )(dy).$$
The last integral is $\geq 0$, and
\begin{eqnarray*}
\int _{\Sigma}\bigl(U^{\mu } (x)+{1\over 2}Q(x)\bigr)\mu^*_\Sigma(dx)
&\geq &C,\cr
\int _{\Sigma}\bigl(U^{\mu } (x)+{1\over 2}Q(x)\bigr)\mu (dx)
&=&C.\cr
\end{eqnarray*}
Therefore $E(\mu^*_\Sigma)\geq E(\mu )$, and $\mu =\mu^*_\Sigma$
by Theorem~\ref{thm3}.

\end{proof}
\section{Examples}
\begin{proposition} For $k<2$, the function
$$V_k(x)=\left\{\begin{aligned}&{\rm sing}(k)|x|^k,\;k\neq 0,\\ &\log|x|,\quad\;\;\;k=0,\end{aligned}\right.$$ is admissible on $\Bbb R^*$.
\end{proposition}
\begin{proof} The function $V_k$ is lower semicontinuous.

Let \( k \in (0, 2) \). The function \( |x|^k \) admits the integral representation:
\[
|x|^k = c(k) \int_0^\infty \left(1 - e^{-x^2 t^2}\right) \frac{dt}{t^{k+1}},
\]
where the constant \( c(k) \) is given by:
\[
c(k) = \frac{k}{\Gamma\left(1 - \frac{k}{2}\right)}.
\]
This representation implies that the function \( V_k(x) = |x|^k \) is negative definite on \( \mathbb{R} \). Consequently, for any compactly supported signed measure \( \mu \) on \( \mathbb{R} \), we have:
\[
\int_{\mathbb{R}^2} V_k(x - y) \, \mu(dx) \, \mu(dy) = -c(k) \int_0^\infty \int_{\mathbb{R}^2} e^{-(x - y)^2 t^2} \, \mu(dx) \, \mu(dy) \, \frac{dt}{t^{k+1}}.
\]
Using the identity:
\[
e^{-(x - y)^2 t^2} = \frac{1}{\sqrt{\pi}} \int_\mathbb{R} e^{-u^2} e^{-2i(x - y)t u} \, du,
\]
we can rewrite the double integral as:
\[
\int_{\mathbb{R}^2} e^{-(x - y)^2 t^2} \, \mu(dx) \, \mu(dy) = \frac{1}{\sqrt{\pi}} \int_\mathbb{R} e^{-u^2} |\hat{\mu}(2t u)|^2 \, du > 0,
\]
where \( \hat{\mu} \) denotes the Fourier transform of \( \mu \). This positivity implies:
\[
\int_{\mathbb{R}^2} V_k(x - y) \, \mu(dx) \, \mu(dy) < 0.
\]

For \( k < 0 \) and \( n \geq 1 \), the function \( \left(|x| + \frac{1}{n}\right)^k \) can be expressed as:
\[
\left(|x| + \frac{1}{n}\right)^k = \frac{1}{\Gamma(-k)} \int_0^\infty e^{-\left(|x| + \frac{1}{n}\right)t} \frac{dt}{t^{k+1}} = \frac{1}{\pi \Gamma(-k)} \int_0^\infty \int_\mathbb{R} e^{i s x t} \frac{e^{-\frac{t}{n}}}{t^{k+1}} \frac{1}{1 + s^2} \, ds \, dt.
\]
Since the measure \( \mu \) is compactly supported, Fubini's theorem allows us to interchange integrals, yielding:
\[
\int_{\mathbb{R}^2} \left(|x - y| + \frac{1}{n}\right)^k \, \mu(dx) \, \mu(dy) = \frac{1}{\pi \Gamma(-k)} \int_0^\infty \int_\mathbb{R} |\hat{\mu}(t s)|^2 \frac{e^{-\frac{t}{n}}}{t^{k+1}} \frac{1}{1 + s^2} \, ds \, dt > 0.
\]
By applying the monotone convergence theorem, we obtain:
\[
\int_{\mathbb{R}^2} V_k(x - y) \, \mu(dx) \, \mu(dy) > 0.
\]

The case \( k = 0 \) is treated by taking the limit as \( k \to 0^- \). First, observe that the function:
\[
k \mapsto \frac{|x - y|^k - 1}{k}
\]
is monotonically increasing on \( (-1, 0) \) for all \( x \neq y \), and satisfies:
\[
\frac{|x - y|^k - 1}{k} \leq \log|x - y|.
\]
The right-hand side is continuous. Thus, for every compactly supported measure \( \mu \) with \( \mu(\mathbb{R}) = 0 \), we have:
\[
\int_{\mathbb{R}^2} \log|x - y| \, \mu(dx) \, \mu(dy) > 0.
\]

\end{proof}
\begin{theorem} Let \( k < 2 \), and let \( Q: \mathbb{R} \to \mathbb{R} \) be a lower semicontinuous function satisfying
\[
\lim_{|x| \to \infty} \left( Q(x) - 4|x|^k \right) = +\infty.
\]
Consider the energy functional
\[
E_k(\mu) = \int_{\mathbb{R}} Q(x) \, \mu(dx) - {\rm sign}(k) \int_{\mathbb{R}^2} |x-y|^k \, \mu(dx) \, \mu(dy),
\]
where \( \mu \) is a probability measure on \( \mathbb{R} \). There exists a unique probability measure \( \mu_k \) that minimizes \( E_k(\mu) \), i.e.,
\[
\inf_{\mu \in \mathcal{M}^1(\mathbb{R})} E_k(\mu) = E_k(\mu_k) = E_k^*.
\]
Moreover, the support of \( \mu_k \) is compact.

\end{theorem}
The uniqueness of the equilibrium measure is not guaranteed for $k=2$.
\begin{proposition}\label{paa}
Let \( Q: \mathbb{R} \to \mathbb{R} \) be a twice differentiable and convex function satisfying
\[
\lim_{|x| \to \infty} \left( Q(x) - 4|x| \right) = +\infty.
\]
The unique probability measure \( \mu^* \) that minimizes the energy functional
\[
E(\mu) = \int_{\mathbb{R}} Q(x) \, \mu(dx) - \int_{\mathbb{R}^2} |x-y| \, \mu(dx) \, \mu(dy)
\]
admits a density \( f \) with respect to the Lebesgue measure given by
\[
f(x) = \frac{1}{2} Q''(x) \chi_{[a,b]}(x),
\]
where \( a \) and \( b \) are constants determined by the conditions
\[
Q'(b) - Q'(a) = 2 \quad \text{and} \quad Q(b) - Q(a) = a + b.
\]
\end{proposition}

\begin{proof}
Assume \( \mu(dx) = f(x) \, dx \). Differentiating the relation from Proposition \eqref{p3} twice yields, for \( x \in [a, b] \),
\[
f(x) = \frac{1}{2} Q''(x).
\]
Since \( \mu^* \) is a probability measure, we have \( Q'(b) = Q'(a) + 2\). For \( x \in [a, b] \), the equilibrium condition implies
\[
\frac{1}{2} Q(x) - \frac{1}{2} \int_a^b |x-y| Q''(y) \, dy = C,
\]
where \( C \) is a constant. Evaluating this at \( x = a \) and \( x = b \), we obtain
\[
\frac{1}{2} Q(b) - \frac{1}{2} \int_a^b (b-y) Q''(y) \, dy = \frac{1}{2} Q(a) - \frac{1}{2} \int_a^b (y-a) Q''(y) \, dy.
\]
Simplifying, we find
\[
\frac{1}{2} (Q(b) - Q(a)) - \frac{b + a}{2} \int_a^b Q''(y) \, dy = 0.
\]
Using \( \int_a^b Q''(y) \, dy = Q'(b) - Q'(a) \), this reduces to
\[
\frac{1}{2} (Q(b) - Q(a)) - \frac{b + a}{2} (Q'(b) - Q'(a)) = 0.
\]
Substituting \( Q'(b) - Q'(a) = 2 \), we conclude
\[
Q(b) - Q(a) = a + b.
\]
This completes the proof.
\end{proof}
\begin{Example}
For \( Q(x) = x^2 + \lambda x + \gamma \), the constants \( a \) and \( b \) satisfy \( b = a + 1 \). Substituting into the conditions \( Q'(b) - Q'(a) = 2 \) and \( Q(b) - Q(a) = a + b \), we obtain
\[
a = \frac{\lambda - 1}{2}, \quad b = \frac{\lambda + 1}{2}.
\]
Thus, the minimizing measure \( \mu^* \) has the density
\[
\mu^*(dx) = \chi_{[(\lambda - 1)/2, (\lambda + 1)/2]}(x) \, dx.
\]
\end{Example}


\section{Generalized Ullman density}
 For \( k \in (-1, 2) \), define the function \( V_k(x) \) as follows:

\[
V_k(x) =
\begin{cases}
\text{sign}(k) |x|^k, & k \neq 0, \\
\log|x|, & k = 0.
\end{cases}
\]

\begin{lemma}\label{li1}

Let \( a < 1 \), \( b > -1 \), and \( c \in \mathbb{R} \), with the condition \( 2a + b - c > 0 \). Define the function \( F(a, b, c) \) as:

\[
F(a, b, c) = \int_{-1}^1 |x|^b \int_1^{1/|x|} \frac{u^c}{(u^2 - 1)^a} \, du \, dx.
\]

Then,

\[
F(a, b, c) = \frac{\Gamma(1 - a) \Gamma\left(\frac{2a + b - c}{2}\right)}{(b + 1) \Gamma\left(\frac{2 + b - c}{2}\right)}.
\]
\end{lemma}
\begin{proof} The proof proceeds as follows:

\[
\begin{aligned}
F(a, b, c) &= \frac{1}{S_{a,b,c}} \int_{-1}^1 |x|^b \int_1^{1/|x|} \frac{u^c}{(u^2 - 1)^a} \, du \, dx \\
&= 2 \int_0^1 x^{b + 2a - c - 1} \int_x^1 \frac{u^c}{(u^2 - x^2)^a} \, du \, dx \\
&= 2 \int_0^1 \int_0^u x^{b + 2a - c - 1} \frac{u^c}{(u^2 - x^2)^a} \, dx \, du \\
&= 2 \int_0^1 \int_0^1 x^{2a + b - c - 1} \frac{u^{\alpha + b}}{(1 - x^2)^a} \, dx \, du \\
&= \frac{1}{\alpha + b + 1} \frac{\Gamma(1 - a) \Gamma\left(\frac{2a + b - c}{2}\right)}{\Gamma\left(\frac{2 + b - c}{2}\right)}.
\end{aligned}
\]

This completes the proof
\end{proof}

\begin{proposition}\label{pa}  \
1) For \( k \in (-1, 1) \), \(k\neq 0\) and \( \beta > 1 \), define the function
\[
f_{k,\beta}(x) = \frac{1}{C_{k,\beta}} |x|^{\beta - k - 1} \int_1^{A_{k,\beta} / |x|} \frac{u^{\beta - 1}}{(u^2 - 1)^{(1 + k)/2}} \, du,
\]
and let \( \mu_{k,\beta} \) be the probability measure supported on \([-A_{k,\beta}, A_{k,\beta}]\) with density \( f_{k,\beta} \), where
\[
C_{k,\beta} = \sqrt{\pi} \frac{|k| \, \Gamma\left(\frac{1 - k}{2}\right) \Gamma\left(\frac{1 + k}{2}\right)}{\Gamma\left(1 - \frac{k}{2}\right)} \frac{\Gamma\left(\frac{\beta - k}{2}\right)}{\beta \, \Gamma\left(\frac{\beta + 1}{2}\right)},
\]
and
\[
A_{k,\beta} = \left( \frac{2 |k| \, \Gamma\left(\frac{1 + k}{2}\right) \Gamma\left(1 + \frac{\beta - k}{2}\right)}{\beta \, \Gamma\left(\frac{\beta + 1}{2}\right)} \right)^{\frac{1}{\beta - k}}.
\]
There exists a constant \( \sigma_{k,\beta} \) such that
\[
\frac{1}{2} |x|^\beta - \int_{\mathbb{R}} V_k(x - y) \, \mu_{k,\beta}(dy) =
\begin{cases}
\sigma_{k,\beta}, & \text{if } |x| \leq A_{k,\beta}, \\
\geq \sigma_{k,\beta}, & \text{if } |x| \geq A_{k,\beta}.
\end{cases}
\]
Moreover,
\[
E^*_{k,\beta} = (A_{k,\beta})^k \frac{(\beta - k)(1 + k - 2\beta)}{\beta \sqrt{\pi} (2\beta - k)} \Gamma\left(\frac{1 + k}{2}\right) \Gamma\left(1 - \frac{k}{2}\right).
\]
For \( k = 0 \),
\[
E^*_\beta = \frac{3}{2\beta} - \log \frac{A_\beta}{2}.
\]

2) For \( k \in (1, 2) \) and \( \beta > k \), define the function
\[
g_{k,\beta}(x) = \frac{1}{C'_{k,\beta}} |x|^{\beta - k - 1} \int_1^{A'_{k,\beta} / |x|} \frac{u^{\beta - 2}}{(u^2 - 1)^{k/2}} \, du,
\]
and let \( \nu_{k,\beta} \) be the probability measure supported on \([-A'_{k,\beta}, A'_{k,\beta}]\) with density \( g_{k,\beta} \), where
\[
C'_{k,\beta} = \sqrt{\pi} \frac{k(k - 1) \, \Gamma\left(\frac{k}{2}\right) \Gamma\left(1 - \frac{k}{2}\right)}{\beta(\beta - 1) \, \Gamma\left(\frac{3 - k}{2}\right)} \frac{\Gamma\left(\frac{\beta - k}{2}\right)}{\Gamma\left(\frac{\beta}{2}\right)},
\]
and
\[
A'_{k,\beta} = \left( \frac{k(k - 1) \, \Gamma\left(\frac{k}{2}\right) \Gamma\left(1 + \frac{\beta - k}{2}\right)}{(\beta - 1) \, \Gamma\left(1 + \frac{\beta}{2}\right)} \right)^{\frac{1}{\beta - k}}.
\]
There exists a constant \( \sigma''_{k,\beta} \) such that
\[
\frac{1}{2} |x|^\beta - \int_{\mathbb{R}} V_k(x - y) \, \nu_{k,\beta}(dy) =
\begin{cases}
\sigma''_{k,\beta}, & \text{if } |x| \leq A'_{k,\beta}, \\
\geq \sigma''_{k,\beta}, & \text{if } |x| \geq A'_{k,\beta}.
\end{cases}
\]
The energy is given by
\[
E^*_{k,\beta} = \frac{\beta - k}{\sqrt{\pi}} (A'_{k,\beta})^k \Gamma\left(\frac{3 - k}{2}\right) \left( \frac{\Gamma\left(\frac{\beta + 1}{2}\right)}{2(2\beta - k) \, \Gamma\left(\frac{\beta - k + 3}{2}\right)} - \frac{2}{\beta \sqrt{\pi}} \Gamma\left(\frac{k + 1}{2}\right) \right).
\]

3) For \( \beta > 1 \), define the function
\[
h_\beta(x) = \frac{\beta(\beta - 1)}{4} |x|^{\beta - 2},
\]
and let \( \mu_\beta \) be the probability measure supported on \(\left[-\left(\frac{2}{\beta}\right)^{1/(\beta - 1)}, \left(\frac{2}{\beta}\right)^{1/(\beta - 1)}\right]\) with density \( h_\beta \). Then
\[
\frac{1}{2} |x|^\beta - \int_{\mathbb{R}} V_1(x - y) \, \mu_\beta(dy) =
\begin{cases}
\frac{1 - \beta}{2}, & \text{if } |x| \leq \left(\frac{2}{\beta}\right)^{1/(\beta - 1)}, \\
\geq \frac{1 - \beta}{2}, & \text{if } |x| \geq \left(\frac{2}{\beta}\right)^{1/(\beta - 1)}.
\end{cases}
\]
Moreover,
\[
E^*_\beta = \frac{\beta - 1}{2} \left( \frac{\beta (2/\beta)^{\frac{2\beta - 1}{\beta - 1}}}{4\beta - 2} - 1 \right).
\]

\end{proposition}
\begin{corollary}

For \(\beta = 2\) and \(k \in (-1, 2)\) with \(k \neq 0\), the equilibrium measure \(\mu^*_k\) admits the density
\[
f_k(x) = \frac{1}{|k| C_k} \left(A_k^2 - x^2\right)^{\frac{1 - k}{2}} \chi_{[-A_k, A_k]}(x),
\]
where the constants \(C_k\) and \(A_k\) are defined as
\[
C_k = \frac{\pi (1 - k)}{\sin\left(\frac{\pi (1 - k)}{2}\right)}, \quad
A_k = \left(\frac{|k| C_k \Gamma\left(\frac{4 - k}{2}\right)}{\sqrt{\pi} \Gamma\left(\frac{3 - k}{2}\right)}\right)^{\frac{1}{2 - k}}.
\]
The associated energy is given by
\[
E_k = \frac{k - 2}{k (4 - k)} A_k^2.
\]

For the case \(k = 0\), the density \(f_0\) of the equilibrium measure \(\mu^*_0\) simplifies to
\[
f_0(x) = \frac{1}{\pi} \sqrt{2 - x^2} \chi_{[-\sqrt{2}, \sqrt{2}]}(x),
\]
with the corresponding energy
\[
E_0 = \frac{3}{4} + \frac{1}{2} \log 2.
\]

\end{corollary}
\begin{lemma}
Let \(\beta > 1\). For every bounded continuous function \(\varphi\) on \(\mathbb{R}\), the following limit holds:

\[
\lim_{k \to 0} \int_{\mathbb{R}} \varphi\left(k^{-\frac{1}{\beta - k}} x\right) \mu_{k, \beta}(dx) = \int_{\mathbb{R}} \varphi(t) \mu_{\beta}(dt),
\]

where \(\mu_{\beta}\) is the Ullman probability measure supported on \([-A_{\beta}, A_{\beta}]\) with density

\[
f_{\beta}(x) = \frac{1}{C_{\beta}} |x|^{\beta - 1} \int_{1}^{A_{\beta} / |x|} \frac{u^{\beta - 1}}{\sqrt{u^2 - 1}} \, du.
\]

Here, the constants \(C_{\beta}\) and \(A_{\beta}\) are defined as:

\[
C_{\beta} = \pi \frac{\Gamma\left(\frac{1}{2}\right) \Gamma\left(\frac{\beta}{2}\right)}{\beta \Gamma\left(\frac{\beta + 1}{2}\right)},
\;\; A_{\beta} = \left(\frac{\Gamma\left(\frac{1}{2}\right) \Gamma\left(\frac{\beta}{2}\right)}{\Gamma\left(\frac{\beta + 1}{2}\right)}\right)^{\frac{1}{\beta}}.
\]

\end{lemma}

\begin{proof}

1) Proof for \( x \in (-1, 1) \):

By the substitutions \( y \leftrightarrow A y \) and \( x \leftrightarrow A x \), we need to show that for \( x \in (-1, 1) \),
\[
\varphi(x) := \int_{-1}^1 \text{sgn}(x - y) |x - y|^{k-1} |y|^{\beta - k - 1} \int_1^{1/|y|} \frac{u^{\beta - 1}}{(u^2 - 1)^{(1 + k)/2}} \, du \, dy = C_{k, \beta} \beta |x|^{\beta - 1}.
\]
Since the density is even, it suffices to prove the result for \( x \in [0, 1) \). Using the change of variable \( z = |y| u \), we obtain
\[
\varphi(x) = \int_{-1}^1 \text{sgn}(x - y) |x - y|^{k-1} \int_{|y|}^1 \frac{z^{\beta - 1}}{(z^2 - y^2)^{(1 + k)/2}} \, dz \, dy.
\]
From [14, (17.23.26), (17.34.10) table of integrals], for \( 0 < k < 2 \),
\[
x^{k-1} = \frac{2 \sin(k \pi / 2) \Gamma(k)}{\pi} \int_0^\infty t^{-k} \sin(x t) \, dt.
\]
Thus, for \( 0 < k < 2 \),
\[
\varphi(x) = \frac{2 \sin(k \pi / 2) \Gamma(k)}{\pi} \int_{-1}^1 \int_0^\infty t^{-k} \sin((x - y) t) \int_{|y|}^1 \frac{z^{\beta - 1}}{(z^2 - y^2)^{(1 + k)/2}} \, dt \, dz \, dy.
\]
Simplifying further, we get
\[
\varphi(x) = \frac{4 \sin(k \pi / 2) \Gamma(k)}{\pi} \int_0^\infty t^{-k} \sin(x t) \int_0^1 \cos(y t) \int_y^1 \frac{z^{\beta - 1}}{(z^2 - y^2)^{(1 + k)/2}} \, dz \, dy \, dt.
\]
By changing the order of integration, we write
\[
\varphi(x) = \frac{4 \sin(k \pi / 2) \Gamma(k)}{\pi} \int_0^\infty t^{-k} \sin(x t) \int_0^1 \int_0^z \frac{\cos(y t)}{(z^2 - y^2)^{(1 + k)/2}} z^{\beta - 1} \, dy \, dz \, dt.
\]
Using the formula [KU (65)6, WA 35(4)a], for \( \nu > -1/2 \),
\[
J_\nu(\xi) = \frac{2 (\xi / 2)^\nu}{\sqrt{\pi} \Gamma(\nu + 1/2)} \int_0^1 (1 - t^2)^{\nu - 1/2} \cos(\xi t) \, dt.
\]
For \( 0 < k < 1 \), we obtain
\[
\int_0^1 \int_0^z \frac{\cos(y t)}{(z^2 - y^2)^{(1 + k)/2}} z^{\beta - 1} \, dy \, dz = 2^{-1 - k/2} \sqrt{\pi} \Gamma\left(\frac{1 - k}{2}\right) t^{k/2} \int_0^1 z^{\beta - k/2 - 1} J_{-k/2}(z t) \, dz.
\]
Thus,
\[
\varphi(x) = h_k \int_0^1 \int_0^\infty t^{-k/2} \sin(x t) J_{-k/2}(z t) z^{\beta - k/2 - 1} \, dt \, dz,
\]
where
\[
h_k = 2^{1 - k/2} \sqrt{\pi} \Gamma\left(\frac{1 - k}{2}\right) \frac{\sin(k \pi / 2) \Gamma(k)}{\pi}.
\]
From [14, (6.699.5)], for \( -1 < \nu < 1/2 \),
\[
\int_0^\infty t^\nu J_\nu(b t) \sin(x t) \, dt =
\begin{cases}
0, & \text{if } 0 < x < b, \\
\frac{\sqrt{\pi} 2^\nu b^\nu}{\Gamma(1/2 - \nu)} (x^2 - b^2)^{-\nu - 1/2}, & \text{if } 0 < b < x.
\end{cases}
\]
For \( k \in (0, 1) \),
\[
\int_0^\infty t^{-k/2} \sin(x t) J_{-k/2}(z t) \, dt =
\begin{cases}
0, & \text{if } 0 < x < z, \\
\frac{\sqrt{\pi} 2^{-k/2} z^{-k/2}}{\Gamma((k + 1)/2)} (x^2 - z^2)^{(k - 1)/2}, & \text{if } 0 < z < x.
\end{cases}
\]
Therefore, for \( x \in (0, 1) \),
\[
\varphi(x) = h_k \frac{\sqrt{\pi} 2^{-k/2}}{\Gamma((k + 1)/2)} \int_0^x z^{\beta - k - 1} (x^2 - z^2)^{(k - 1)/2} \, dz.
\]
This simplifies to
\[
\varphi(x) = c_{k, \beta} x^{\beta - 1},
\]
where
\[
c_{k, \beta} = \frac{\pi}{2^k} \Gamma\left(\frac{1 - k}{2}\right) \Gamma(k) \frac{\Gamma\left(\frac{\beta - k}{2}\right)}{\Gamma\left(\frac{1 + k}{2}\right) \Gamma\left(\frac{\beta + 1}{2}\right)} \frac{\sin(k \pi / 2)}{\pi}.
\]
Using the Lagrange formula for the gamma function, we obtain
\[
c_{k, \beta} = \sqrt{\pi} \frac{\Gamma\left(\frac{1 - k}{2}\right) \Gamma\left(\frac{1 + k}{2}\right)}{\Gamma\left(1 - \frac{k}{2}\right)} \frac{2 \Gamma\left(\frac{\beta - k}{2}\right)}{\Gamma\left(\frac{\beta + 1}{2}\right)}=\frac{\beta}{2|k|} C_{k, \beta}.
\]
For \( x \in (0, 1) \),
\[
\varphi(x) =  \frac{\beta}{2|k|} C_{k, \beta} x^{\beta - 1}.
\]
For \( x > 1 \),
\[
\varphi(x) = h_k \frac{\sqrt{\pi} 2^{-k/2}}{\Gamma((k + 1)/2)} \int_0^1 z^{\beta - k - 1} (x^2 - z^2)^{(k - 1)/2} \, dz.
\]
Thus,
\[
\varphi(x) = \frac{\beta}{2|k|} C_{k, \beta} x^{\beta - 1} - h_k \frac{\sqrt{\pi} 2^{-k/2}}{\Gamma((k + 1)/2)} x^{\beta - 1} \int_0^{1/x} z^{\beta - k - 1} (1 - z^2)^{(k - 1)/2} \, dz.
\]
We conclude that
\[
\frac{\beta}{2|k|} C_{k, \beta} x^{\beta - 1} - \varphi(x) =
\begin{cases}
0, & \text{if } 0 < x < 1, \\
\theta_{k, \beta}(x), & \text{if } x > 1,
\end{cases}
\]
where
\[
\theta_{k, \beta}(x) = h_k \frac{\sqrt{\pi} 2^{-k/2}}{\Gamma((k + 1)/2)} x^{\beta - 1} \int_0^{1/x} z^{\beta - k - 1} (1 - z^2)^{(k - 1)/2} \, dz.
\]
Integrate the relation above, we get the desired result.

By the Lagrange formula, we have
\[
c_{k, \beta} = \frac{\sqrt{\pi}}{2} \frac{\Gamma\left(\frac{1 - k}{2}\right) \Gamma\left(\frac{1 + k}{2}\right)}{\Gamma\left(1 - \frac{k}{2}\right)} \frac{\Gamma\left(\frac{\beta - k}{2}\right)}{\Gamma\left(\frac{\beta + 1}{2}\right)}.
\]
Therefore,
\[
c_\beta := \lim_{k \to 0} c_{k, \beta} = \frac{\pi}{2} \frac{\Gamma\left(\frac{1}{2}\right) \Gamma\left(\frac{\beta}{2}\right)}{\Gamma\left(\frac{\beta + 1}{2}\right)}.
\]
By analytic continuation, we obtain
\[
c_\beta x^{\beta - 1} - \int_{-1}^1 \frac{1}{x - y} |y|^{\beta - 1} \int_1^{1/|y|} \frac{u^{\beta - 1}}{(u^2 - 1)^{1/2}} \, du \, dy =
\begin{cases}
0, & \text{if } 0 < x < 1, \\
\theta_\beta(x), & \text{if } x > 1,
\end{cases}
\]
where
\[
\theta_\beta(x) = \sqrt{\pi} x^{\beta - 1} \int_0^{1/x} \frac{z^{\beta - 1}}{(1 - z^2)^{1/2}} \, dz.
\]

2) Proof for \( k \in (1, 2) \):

We follow a similar approach as in the first case. For \( k \in (1, 2) \), \( \beta > 2 \), and \( x > 0 \), we define
\[
\psi(x) := \int_{-1}^1 |x - y|^{k - 2} |y|^{\beta - k - 1} \int_1^{1/|y|} \frac{u^{\beta - 2}}{(u^2 - 1)^{k/2}} \, du \, dy.
\]
Using the symmetry of the density, we focus on \( x > 0 \). From [14, (17.23.26)], for \( k \in (1, 3) \),
\[
\psi(x) = \frac{2 \sin((k - 1) \pi / 2) \Gamma(k - 1)}{\pi} \int_{-1}^1 \int_0^\infty t^{1 - k} \sin((x - y) t) |y|^{\beta - k - 1} \int_1^{1/|y|} \frac{u^{\beta - 2}}{(u^2 - 1)^{k/2}} \, du \, dt \, dy.
\]
By changing the order of integration, we obtain
\[
\psi(x) = \frac{4 \sin((k - 1) \pi / 2) \Gamma(k - 1)}{\pi} \int_0^1 \int_0^\infty t^{1 - k} \sin(x t) \cos(y t) y^{\beta - k - 1} \int_1^{1/y} \frac{u^{\beta - 2}}{(u^2 - 1)^{k/2}} \, du \, dt \, dy.
\]
Using the identity for Bessel functions, we simplify this to
\[
\psi(x) = \delta_k \int_0^1 u^{\beta - (k + 3)/2} \int_0^\infty t^{(1 - k)/2} \sin(x t) J_{(1 - k)/2}(t u) \, dt \, du,
\]
where
\[
\delta_k = 2^{(3 - k)/2} \sqrt{\pi} \Gamma(k - 1) \Gamma\left(1 - \frac{k}{2}\right) \frac{\sin((k - 1) \pi / 2)}{\pi}.
\]
From [14, (6.699.5)], for \( -1 < \nu < 1/2 \),
\[
\int_0^\infty t^\nu J_\nu(b t) \sin(x t) \, dt =
\begin{cases}
0, & \text{if } 0 < x < b, \\
\frac{\sqrt{\pi} 2^\nu b^\nu}{\Gamma(1/2 - \nu)} (x^2 - b^2)^{-\nu - 1/2}, & \text{if } 0 < b < x.
\end{cases}
\]
Thus, for \( x \in (0, 1) \),
\[
\psi(x) = \delta_k \frac{\sqrt{\pi} 2^{(1 - k)/2}}{\Gamma(k / 2)} \int_0^x u^{\beta - k - 1} (x^2 - u^2)^{k/2 - 1} \, du.
\]
This simplifies to
\[
\psi(x) = b_{k, \beta} x^{\beta - 2},
\]
where
\[
b_{k, \beta} = \frac{\sqrt{\pi}}{2} \frac{\Gamma\left(\frac{k}{2}\right) \Gamma\left(1 - \frac{k}{2}\right)}{ \Gamma\left(\frac{3 - k}{2}\right)} \frac{\Gamma\left(\frac{\beta - k}{2}\right)}{\Gamma\left(\frac{\beta}{2}\right)}=\frac12\frac{\beta (\beta - 1)C'_{k,\beta}}{k(k-1)}.
\]

For \( x > 1 \), we have
\[
\psi(x) = b_{k, \beta} x^{\beta - 2} - g_{k, \beta}(x),
\]
where
\[
g_{k, \beta}(x) = \delta_k \frac{\sqrt{\pi} 2^{(1 - k)/2}}{\Gamma(k / 2)} x^{\beta - 2} \int_{1/x}^1 u^{\beta - k - 1} (1 - u^2)^{k/2 - 1} \, du.
\]
Integrating, then there exits a constant \(\sigma''_{k, \beta}\in\Bbb R\) such that
\[
\frac{1}{2} x^\beta - \int_\mathbb{R} |x - y|^k \nu_{k, \beta}(dy) =
\begin{cases}
\sigma''_{k, \beta}, & \text{if } 0 < x < 1, \\
\geq \sigma''_{k, \beta}, & \text{if } x > 1.
\end{cases}
\]
The energy is given by
\[
E^*_{k, \beta} = \frac{1}{2} \int_\mathbb{R} |x|^\beta \nu_{k, \beta}(dy) + \sigma''_{k, \beta}.
\]

3) Proof for \( x \in (-1, 1) \):

By a change of variable, it suffices to prove the result for \( x \in (-1, 1) \) with a measure supported on \( (-1, 1) \). By symmetry, we assume \( x \in (0, 1) \).

For \( 0 < x < 1 \),
\[
\int_{-1}^1 |x - y| |y|^{\beta - 2} \, dy = \int_0^1 (|x - y| + |x + y|) |y|^{\beta - 2} \, dy = \frac{2}{\beta (\beta - 1)} x^\beta + \frac{2}{\beta}.
\]
For \( x \geq 1 \),
\[
\int_{-1}^1 |x - y| |y|^{\beta - 2} \, dy = \int_{-1}^1 (x - y) |y|^{\beta - 2} \, dy = 2 x \int_0^1 |y|^{\beta - 2} \, dy = \frac{2}{\beta - 1} x.
\]
Thus,
\[
\frac{1}{2} x^\beta - \int_\mathbb{R} |x - y| \mu(dy) =
\begin{cases}
\frac{1 - \beta}{2}, & \text{if } 0 < x < 1, \\
(x^{\beta - 1} - \beta) x \geq \frac{1 - \beta}{2}, & \text{if } x > 1.
\end{cases}
\]
The energy is given by
\[
E^*_\beta = \frac{\beta - 1}{2} \left( \frac{\beta (2 / \beta)^{(2 \beta - 1) / (\beta - 1)}}{4 \beta - 2} - 1 \right).
\]

\end{proof}

\section {Density of positive particles}\label{den}

Let us define the set
\[
{\cal M}^1_\omega(\Bbb R) = \{\mu \in {\cal M}^1(\Bbb R) \mid \mu([\omega, +\infty)) = 1\}.
\]

\begin{theorem}\label{th0}
Let \( k \in (-1, 1) \). Consider the energy functional
\[
E_k(\mu) = \int_{\Bbb R} x^2 \mu(dx) - \int_{\Bbb R^2} V_k(x - y) \mu(dx) \mu(dy).
\]
There exists a unique probability measure \( \mu^*_{k, \omega} \) on \( \Bbb R \) such that
\[
\inf_{\mu \in {\cal M}^1_\omega(\Bbb R)} E_k(\mu) = E_k(\mu^*_{k, \omega}) = E^*_{k, \omega}.
\]
Moreover, the support of \( \mu^*_{k, \omega} \) is compact.
\end{theorem}

To derive the analytic expression of the measure \( \mu^*_{k, \omega} \), we first state a preliminary result.

\begin{proposition}\label{pi}
For \( k \in (-1, 1) \), \( k \neq 0 \), and \( \alpha_k > 0 \), define the function
\[
f(x, \omega) = (3 - k)x^{2 - k} + 2(2 - k)\omega x^{1 - k} - \alpha_k
\]
on \( (0, +\infty) \times \Bbb R \). The following properties hold:

1. For every \( \omega \in \Bbb R \), there exists a unique real number \( L_k(\omega) > 0 \) such that \( f(L_k(\omega), \omega) = 0 \).

2. For every \( k \), the function \( \omega \mapsto L_k(\omega) \) is continuously differentiable on \( \Bbb R \), and the function \( \omega \mapsto L_k(\omega) + 2\omega \) is strictly increasing.

3. \( L_k(\omega) + 2\omega > 0 \) if and only if \( \omega > \omega_k \), where
\[
\omega_k = -\frac{1}{2} (\alpha_k)^{\frac{1}{2 - k}}.
\]
\end{proposition}

\begin{proof}
We begin by analyzing the partial derivative of \( f \) with respect to \( x \):

\[
\frac{\partial f}{\partial x}(x, \omega) = (2 - k)\left((3 - k)x^{1 - k} + 2(1 - k)\omega x^{-k}\right).
\]

Case 1: \(\omega \geq 0\).
For \( x > 0 \) and \( k \in (-1, 1) \) with \( k \neq 0 \), we have \(\frac{\partial f}{\partial x}(x, \omega) > 0\). Additionally, \( f(0, \omega) = -\alpha_k < 0 \).

Case 2: \(\omega < 0\).
For \( x > -\frac{2(1 - k)}{3 - k}\omega = x_0 > 0 \), we have \(\frac{\partial f}{\partial x}(x, \omega) > 0\). Thus, the function \( x \mapsto f(x, \omega) \) is strictly increasing on \( (x_0, +\infty) \). Since \(\lim_{x \to +\infty} f(x, \omega) = +\infty\) and

\[
f(x_0, \omega) = -(3 - k)x_0^{2 - k} - \alpha_k < 0,
\]

it follows that for every \(\omega \in\Bbb R\) and \( k \in (-1, 1) \), there exists a unique \( L_k(\omega) > 0 \) such that \( f(L_k(\omega), \omega) = 0 \).

2) Implicit Function Theorem and Monotonicity:

Fix \(\omega_0 \in \mathbb{R}\), and let \( L_k(\omega_0) \) be the unique solution as described above. The function \( (x, \omega) \mapsto f(x, \omega) \) is continuously differentiable on \( (0, +\infty) \times \mathbb{R} \), and

\[
\frac{\partial f}{\partial x}(L_k(\omega_0), \omega_0) = (2 - k)\left((3 - k)L_k(\omega_0)^{1 - k} + 2(1 - k)\omega_0 L_k(\omega_0)^{-k}\right).
\]

From the definition of \( L_k(\omega_0) \), we deduce:

\(\bullet\) If \(\omega_0 \geq 0\), then \(\frac{\partial f}{\partial x}(L_k(\omega_0), \omega_0) > 0\).

\(\bullet\) If \(\omega_0 < 0\), then
  \[
  \frac{\partial f}{\partial x}(L_k(\omega_0), \omega_0) = (2 - k)\frac{\alpha_k - 2\omega_0 L_k(\omega_0)^{1 - k}}{L_k(\omega_0)} > 0.
  \]

By the implicit function theorem, there exists an open set \( U \subset \mathbb{R} \) containing \(\omega_0\) and a unique continuously differentiable function \( g: U \to (0, +\infty) \) such that \( g(\omega_0) = L_k(\omega_0) \) and \( f(g(\omega), \omega) = 0 \) for all \(\omega \in U\). By uniqueness, \( g(\omega) = L_k(\omega) \). Moreover,

\[
g'(\omega)\frac{\partial f}{\partial x}(g(\omega), \omega) + \frac{\partial f}{\partial \omega}(g(\omega), \omega) = 0.
\]

Thus,

\[
L_k'(\omega) = g'(\omega) = -\frac{2(2 - k)L_k(\omega)^{1 - k}}{\frac{\partial f}{\partial x}(L_k(\omega), \omega)} < 0.
\]

This shows that the function \(\omega \mapsto L_k(\omega)\) is strictly decreasing on \(\mathbb{R}\).

3) Behavior of \( L_k(\omega) \) and Existence of \(\omega_k\):

From the equation \( f(L_k(\omega), \omega) = 0 \), we derive:

\[
(3 - k)(L_k(\omega) + 2\omega) = \frac{\alpha_k}{L_k(\omega)^{1 - k}} + 2\omega.
\]

The right-hand side is a strictly increasing function on \(\mathbb{R}\), and since \( k < 3 \), the function \(\theta(\omega) = L_k(\omega) + 2\omega\) is also strictly increasing. Using the inequality \( L_k(\omega) > -\frac{2(1 - k)}{3 - k}\omega \) for \(\omega < 0\), we find that \(\lim_{\omega \to -\infty} L_k(\omega) = +\infty\). Consequently, from the equation above, \(\lim_{\omega \to -\infty} \theta(\omega) = -\infty\). Since \(\theta(0) = L_k(0) = \left(\frac{\alpha_k}{3 - k}\right)^{1/(2 - k)} > 0\), there exists a unique \(\omega_k < 0\) such that

\[
L_k(\omega_k) + 2\omega_k = 0.
\]

Substituting into the equation \( f(L_k(\omega), \omega) = 0 \), we obtain:

\[
\omega_k = -\frac{1}{2} \alpha_k^{1/(2 - k)}.
\]
This completes the proof.
\end{proof}

In the following theorem, we derive an explicit computation of the density of the measure \(\mu^*_{k, \omega}\).

\begin{theorem}\label{tii}
For \(k \in (-1, 1)\), \(k \neq 0\), define the function
\[
g_{k,\omega}(t) = \frac{1}{|k| C_k} \left(L_k(\omega) + \omega - t\right)^{\frac{1-k}{2}} (t - \omega)^{-\frac{1+k}{2}} \left(t - k\omega + \frac{1-k}{2} L_k(\omega)\right),
\]
and \(\nu_{k,\omega}\) the probability measure with density \(g_{k,\omega}\) supported on the interval \([\omega, \omega + L_k(\omega)]\). Here,
\[
C_k = \frac{\pi(1-k)}{\sin\left(\frac{(1-k)\pi}{2}\right)},
\]
and \(L_k(\omega)\) is the unique positive solution to the equation in the variable \(x\):
\[
(3-k)x^{2-k} + 2(2-k)\omega x^{1-k} = \frac{2^{2-k} |k| C_k \Gamma\left(\frac{4-k}{2}\right)}{\sqrt{\pi} \Gamma\left(\frac{3-k}{2}\right)}.
\]

There exists a constant \(S_{k,\omega}\) such that
\[
\frac{1}{2}x^2 - {\rm sgn}(k) \int_{\mathbb{R}} |x-y|^k \nu_{k,\omega}(dy)
\begin{cases}
= S_{k,\omega}, & \text{if } \omega \leq x \leq \omega + L_k(\omega), \\
\geq S_{k,\omega}, & \text{if } x \geq \omega + L_k(\omega).
\end{cases}
\]
Moreover, for \(k \neq 0\),
\[
\begin{aligned}
E^*_{k,\omega} &= \frac{1}{2}\left(\omega + L_k(\omega)\right)^2 - \frac{1+k}{2k}\left(\omega + \frac{3}{4}L_k(\omega)\right) \\
&+  \Bigg(L_k(\omega)^3 \frac{k^2 - 6k + 5}{8(k^2 - 6k + 8)} + \left((3-k)\omega + \frac{1-k}{2} L_k(\omega)\right) \frac{k-1}{4(k-2)} L_k(\omega)^2 \\
&+ \left((3-2k)\omega + (1-k)L_k(\omega)\right) \omega L_k(\omega) \frac{k-1}{2(k-2)} + \omega^2 \left((1-k)\omega + \frac{1-k}{2} L_k(\omega)\right)\Bigg) \\
&\times \frac{\sqrt{\pi} 2^{k-2} \Gamma\left(\frac{1}{2} - \frac{k}{2}\right)}{|k|C_k\Gamma\left(1 - \frac{k}{2}\right)}L_k(\omega)^{1-k}.\\
&=\frac{1}{2}\left(\omega + L_k(\omega)\right)^2 - \frac{1-k^2}{2k(2-k)(2(2-k)\omega+(3-k)L_k(\omega))}\left(\omega + \frac{3}{4}L_k(\omega)\right) \\
&+\Bigg(\frac{5-k}{8(4-k)}L_k(\omega)^3  + \frac{1}{4}\left((3-k)\omega + \frac{1-k}{2} L_k(\omega)\right)  L_k(\omega)^2 \\
&+ \frac{\omega }{2}\Big((3-2k)\omega + (1-k)L_k(\omega)\Big) L_k(\omega)  +(2-k)\omega^2 \left(\omega + \frac{1}{2} L_k(\omega)\right)\Bigg) \\
\end{aligned}
\]

For \(k = 0\),
\[
g_{0,\omega}(t) = \frac{1}{2\pi} \sqrt{\frac{L_0(\omega) - t}{t}} (2t + 2\omega + L_0(\omega)) \chi_{[0, L_0(\omega)]},
\]
and
\[
L_0(\omega) = \frac{2}{3}\left(\sqrt{\omega^2 + 6} - \omega\right).
\]
\end{theorem}

\subsection{Limit Cases}
Let \(k \in (-1, 1)\), \(k \neq 0\). Define
\[
\omega_k = \left(\frac{|k| C_k \Gamma\left(\frac{4-k}{2}\right)}{\sqrt{\pi} \Gamma\left(\frac{3-k}{2}\right)}\right)^{\frac{1}{2-k}}.
\]
If the barrier is located at \(\omega\) with \(\omega > -\omega_k\), the density is given by \(g_{k, \omega}\).

If \(\omega \leq -\omega_k\), then by Proposition \ref{pi}, \(L_k(\omega_k) = 2\omega_k\), and the density \(g_{k,\omega_k}\) is
\[
g_{k,\omega_k}(t) = \frac{1}{|k| C_k} (\omega_k^2 - t^2)^{\frac{1-k}{2}} \chi_{[-\omega_k, \omega_k]}(t).
\]
In this case, the energy is given by
\[
E^*_k = \frac{k-2}{k(4-k)} \omega_k^2.
\]

For \(k = 0\), if the barrier is located at \(\omega\) with \(\omega \geq -\sqrt{2}\), the density is
\[
h_0(t) = \frac{1}{2\pi} \sqrt{\frac{\frac{2}{3}\sqrt{6} - t}{t}} \left(2t + \frac{2}{3}\sqrt{6}\right) \chi_{[0, \frac{2}{3}\sqrt{6}]}.
\]
In this case, the energy is given by \[E_{k, +}=\frac34+\frac12\log 6,\]
If \(\omega \leq -\sqrt{2}\), we obtain the semicircle law
\[
h_{0}(t) = \frac{1}{\pi} \sqrt{2 - t^2},
\]
and \(E_0 = \frac{3}{4} + \frac{1}{2}\log 2\).

\subsection{Special Case:}
Setting \(\omega = 0\) in the previous theorem yields the density for the case where all particles lie on the positive axis, in that case the density is \(h_0\).

\begin{corollary}\label{c1}
For \(k \in (-1, 1)\), \(k \neq 0\), define the function
\[
h_k(t) = \frac{1}{|k| C_k} (a_k - t)^{\frac{1-k}{2}} t^{-\frac{1+k}{2}} \left(t + a_k \frac{1-k}{2}\right),
\]
where
\[
C_k = \frac{\pi(1-k)}{\sin\left(\frac{(1-k)\pi}{2}\right)},
\]
and
\[
a_k = 2(3-k)^{-\frac{1}{2-k}} \omega_k.
\]
Let \(\mu_k\) be the probability measure supported on \([0, a_k]\) with density \(h_k\). There exists a constant \(\sigma_k\) such that
\[
\frac{1}{2}x^2 - {\rm sgn}(k) \int_{\mathbb{R}} |x-y|^k \mu_k(dy)
\begin{cases}
= \sigma_k, & \text{if } 0 < x \leq a_k, \\
\geq \sigma_k, & \text{if } x \geq a_k.
\end{cases}
\]
Moreover, for \(k \neq 0\),
\[
E^*_{k, +} = \frac{k-2}{k(4-k)(3-k)^{\frac{k}{2-k}}} \omega_k^2.
\]
For \(k = 0\),
\[
h_0(t) = \frac{1}{2\pi} \sqrt{\frac{\frac{2}{3}\sqrt{6} - t}{t}} \left(2t + \frac{2}{3}\sqrt{6}\right) \chi_{[0, \frac{2}{3}\sqrt{6}]},
\]
and
\[
E^*_{0, +} = \frac{3}{4} + \frac{1}{2}\log 6.
\]
\end{corollary}

\begin{proposition}
For every \(k \in (0, 2)\), the probability \(p_{n, k}\) that all particles are positive is given by
\[
\lim_{n \to \infty} \frac{1}{n^{2+k}} \log p_{n,k} = \frac{\left(1 - (3-k)^{-\frac{k}{2-k}}\right)(2-k)}{k(4-k)} \omega_k^2.
\]
For \(k = 0\),
\[
\lim_{n \to \infty} \frac{1}{n^{2}} \log p_{n, 0} = \frac{1}{2}\log 3.
\]
\end{proposition}

\begin{remark}
The following limits hold:
\begin{enumerate}
\item
\[
\lim_{k \to 0} \frac{E^*_k + 1 + \frac{1}{2}k \log k}{k} = E_0^* = \frac{3}{4} + \frac{1}{2}\log 2.
\]
\item
\[
\lim_{k \to 0} \frac{E^*_{k, +} + 1 + \frac{1}{2}k \log k}{k} = E_{0, +}^* = \frac{3}{4} + \frac{1}{2}\log 6.
\]
\item
\[
\lim_{k \to 0} \lim_{n \to \infty} \frac{1}{k n^{2+k}} \log p_{n,k} = \frac{1}{2}\log 3.
\]
\end{enumerate}
\end{remark}
To prove Theorem \ref{tii}, we require the following lemma:

\begin{lemma}
For \(k \in (-1, 1)\), \(k \neq 0\), and \(x \in (0, 1)\),
\[
\int_{0}^1 {\rm sign}(x-t) |x-t|^{k-1} (1-t)^{\frac{1-k}{2}} t^{-\frac{1+k}{2}} \left(t + \frac{1-k}{2}\right) dt = \frac{\pi(1-k)}{\sin\left(\frac{1-k}{2}\pi\right)} x,
\]
and
\[
\int_{0}^1 {\rm sign}(x-t) |x-t|^{k-1} (1-t)^{\frac{1-k}{2}} t^{-\frac{1+k}{2}} dt = \frac{\pi}{\sin\left(\frac{1-k}{2}\pi\right)}.
\]
\end{lemma}

\begin{proof}
From \cite{1}, Corollary 2.2, we have for every integers \(m, n, r\) with \(m + r + n + 1 > 0\), \(\nu \in (0, 1)\), and \(x \in (-1, 1)\),
\[
\int_{-1}^1 \frac{{\rm sign}(x-t)}{|x-t|^\nu} (1-t)^{\nu/2 + r}(1+t)^{\nu/2 + n}P_m^{\nu/2 + r, \nu/2 + n}(t) \, dt = v_{n, m, r, \nu} P_{m + r + n + 1}^{\nu/2 - r - 1, \nu/2 - n - 1}(x),
\]
where \(P_m^{\alpha, \beta}\) are the Jacobi polynomials and
\[
v_{n, m, r, \nu} = \frac{\pi (-1)^r 2^{r + n + 1} \Gamma(m + \nu)}{m! \Gamma(\nu) \sin\left(\frac{\nu \pi}{2}\right)}.
\]

Let \(y \in (-1, 1)\). Then,
\[
2^{1+\nu+r+n} \int_{0}^1 \frac{{\rm sign}(y + 1 - 2t)}{|y + 1 - 2t|^\nu} (1-t)^{\nu/2 + r}t^{\nu/2 + n}P_m^{\nu/2 + r, \nu/2 + n}(2t - 1) \, dt = v_{n, m, r, \nu} P_{m + r + n + 1}^{\nu/2 - r - 1, \nu/2 - n - 1}(y).
\]

For \(x = \frac{y + 1}{2} \in (0, 1)\),
\[
2^{1+r+n} \int_{0}^1 \frac{{\rm sign}(x - t)}{|x - t|^\nu} (1-t)^{\nu/2 + r}t^{\nu/2 + n}P_m^{\nu/2 + r, \nu/2 + n}(2t - 1) \, dt = v_{n, m, r, \nu} P_{m + r + n + 1}^{\nu/2 - r - 1, \nu/2 - n - 1}(2x - 1).
\]


For \(k \in (0, 1)\), \(\nu = 1 - k\), \(n = -1\), \(r = 0\), and \(m = 1\), we get
\[
\int_{0}^1 {\rm sign}(x - t) |x - t|^{k - 1} (1 - t)^{\frac{1 - k}{2}} t^{-\frac{1 + k}{2}} P_1^{\frac{1 - k}{2}, -\frac{1 + k}{2}}(2t - 1) \, dt = \frac{\pi(1 - k)}{\sin\left(\frac{1 - k}{2} \pi\right)} P_1^{-\frac{1 + k}{2}, \frac{1 - k}{2}}(2x - 1).
\]

Since,
\[
P_1^{\frac{1 - k}{2}, -\frac{1 + k}{2}}(2t - 1) = \frac{3 - k}{2} + (2 - k)(t - 1) = (2 - k)t + \frac{k - 1}{2},
\]
and
\[
P_1^{-\frac{1 + k}{2}, \frac{1 - k}{2}}(2x - 1) = \frac{1 - k}{2} + (2 - k)(x - 1) = (2 - k)x + \frac{k - 3}{2}.
\]

Thus,
\[
\int_{0}^1 {\rm sign}(x - t) |x - t|^{k - 1} (1 - t)^{\frac{1 - k}{2}} t^{-\frac{1 + k}{2}} \left((2 - k)t + \frac{k - 1}{2}\right) dt = \frac{\pi(1 - k)}{\sin\left(\frac{1 - k}{2} \pi\right)} \left((2 - k)x + \frac{k - 3}{2}\right).
\]

From this, we deduce that
\[
\int_{0}^1 {\rm sign}(x - t) |x - t|^{k - 1} (1 - t)^{\frac{1 - k}{2}} t^{-\frac{1 + k}{2}} dt = \frac{\pi}{\sin\left(\frac{1 - k}{2} \pi\right)},
\]
and
\[
\int_{0}^1 {\rm sign}(x - t) |x - t|^{k - 1} (1 - t)^{\frac{1 - k}{2}} t^{-\frac{1 + k}{2}} \left(t + \frac{1 - k}{2}\right) dt = \frac{\pi(1 - k)}{\sin\left(\frac{1 - k}{2} \pi\right)} x.
\]
This completes the proof of the lemma.
\end{proof}

\begin{proof}[Proof of Theorem \ref{tii}.]
For $x\in(\omega, \omega+L_k(\omega))$, by the substitution $t=\omega+y L_k(\omega)$ and $x=\omega+zL_k(\omega)$, we get
$$\begin{aligned}&|k|\int_0^\infty|x-t|^{k-1}g_{\omega,k}(t)dt\\
&=\frac1{C_{k}}\int_{0}^{1}{\rm sign}(z-y)|z-y|^{k-1}(1-y)^{\frac{1-k}2}y^{-\frac{1+k}2}((y+\frac{1-k}2)L_k(\omega)+(1-k)\omega)dy\\
&=zL_k(\omega)+\omega=x\\
\end{aligned}$$
Therefore,  for $x\in(\omega, \omega+L_k(\omega))$
$${\rm sgn}(k)\int_0^\infty|x-t|^{k}g_{\omega,k}(t)dt=\frac12 x^2+S_{k,\omega}.$$
Now, for $x\geq \omega+L_k(\omega)$, let
$$\psi(x)=x-|k|\int_0^\infty(x-t)^{k-1}\nu_{\omega,k}(t)dt.$$
We have, $$\psi'(x)=1+(1-k)|k|\int_0^\infty(x-t)^{k}g_{\omega,k}(t)dt\geq 0.$$
and $$\psi(x)=\psi(\omega+L_k(\omega))+\int_{\omega+L_k(\omega)}^x\psi'(t)dt.$$
Since,
$\psi(\omega+L_k(\omega))=0$
Whence,
$$x-|k|\int_{\omega}^\infty(x-t)^{k-1}\nu_{k,\omega}(dt)=\int_{\omega+L_k(\omega)}^x\psi'(t)dt,$$
and for $x\geq \omega+L_k(\omega)$
$$\frac12x^2-{\rm sgn}(k)\int_{\omega}^\infty(x-t)^{k}\nu_{k,\omega}(dt)\geq S_{k,\omega}.$$
With
$$\begin{aligned}S_{k,\omega}&=\frac12\Big(\omega+L_k(\omega)\Big)^2-{\rm sgn}(k)\int_{\omega}^\infty(\omega+L_k(\omega)-t)^{k}\nu_{k,\omega}(dt)\\
&=\frac12\Big(\omega+L_k(\omega)\Big)^2-\frac1{kC_{k}}\int_{\omega}^{\omega+L_k(\omega)}
(\omega+L_k(\omega)-t)^{\frac{1+k}2}(t-\omega)^{-\frac{1+k}2}(t-k\omega+\frac{1-k}{2} L_k(\omega))dt\\
&=\frac12\Big(\omega+L_k(\omega)\Big)^2-\frac1{kC_{k}}\int_{0}^{L_k(\omega)}
(L_k(\omega)-t)^{\frac{1+k}2}t^{-\frac{1+k}2}(t+(1-k)\omega+\frac{1-k}{2} L_k(\omega))dt\\
&=\frac12\Big(\omega+L_k(\omega)\Big)^2-\frac{L_k(\omega)}{kC_{k}}\int_{0}^{1}
(1-t)^{\frac{1+k}2}t^{-\frac{1+k}2}(tL_k(\omega)+(1-k)\omega+\frac{1-k}{2} L_k(\omega))dt\\
&=\frac12\Big(\omega+L_k(\omega)\Big)^2-\frac{L_k(\omega)}{k}(\frac18(1+k)L_k(\omega)+\frac12(1+k)(\omega+\frac{1}{2} L_k(\omega)))\\
&=\frac12\Big(\omega+L_k(\omega)\Big)^2-\frac{1+k}{2k}(\omega+\frac34L_k(\omega))L_k(\omega).\end{aligned}$$
Moreover,
$$E^*_{k,\omega}=\frac12\int_\Bbb R t^2\nu_{k,\omega}(dt)+S_{k,\omega}.$$
Since,
$$\begin{aligned}\frac12\int_\Bbb R t^2\nu_{k,\omega}(dt)&=
\frac1{2kC_{k}}\int_{\omega}^{\omega+L_k(\omega)}
(\omega+L_k(\omega)-t)^{\frac{1-k}2}(t-\omega)^{-\frac{1+k}2}t^2(t-k\omega+\frac{1-k}{2} L_k(\omega))dt\\
&=\frac1{2kC_{k}}\int_0^{L_k(\omega)}(L_k(\omega)-t)^{\frac{1-k}2}t^{-\frac{1+k}2}(t+\omega)^2(t+(1-k)\omega+\frac{1-k}{2} L_k(\omega))dt\\
&=\frac{L_k(\omega)^{1-k}}{2kC_{k}}\int_0^{1}(1-t)^{\frac{1-k}2}t^{-\frac{1+k}2}(t L_k(\omega)+\omega)^2(tL_k(\omega)+(1-k)\omega+\frac{1-k}{2} L_k(\omega))dt\\
&=\frac{L_k(\omega)^{1-k}}{2kC_{k}}\Bigl(L_k(\omega)^3\frac{k^2-6 k+5}{8 \left(k^2-6 k+8\right)}+((3-k)\omega+\frac{1-k}{2} L_k(\omega))\frac{k-1}{4 (k-2)}L_k(\omega)^2\\
&+((3-2k)\omega+(1-k) L_k(\omega))\omega L_k(\omega))\frac{k-1}{2 (k-2)}+\omega^2((1-k)\omega+\frac{1-k}{2} L_k(\omega))\Bigr)\\
&\times\frac{\sqrt{\pi } 2^{k-1} \Gamma \left(\frac{1}{2}-\frac{k}{2}\right)}{\Gamma \left(1-\frac{k}{2}\right)}
\end{aligned}$$
From the equation $$(3-k)L_k^{2-k}+2(2-k)w L_k^{1-k}=2^{2-k}|k|C_k\frac{\Gamma(4/2 - k/2)}{\sqrt\pi\Gamma(3/2 - k/2)},$$
we get
$$\begin{aligned}&128 C_k^4 (1-k)^3 kE^*_{k,\omega}\\&=
16 C_k^2 (k-1) \left(-4 \omega^2 k+4 \omega L_k(\omega) (1-2 k)+(L_k(\omega))^2 (k (3 (k-2) k-8)+10)\right)\\&+2^kL_k(\omega) \Bigl(8 \omega^3 (4-k) (2-k)+4 (4-k) (5-2 k)\omega^2L_k(\omega)
+2 (4-k) (5-(3-k) k)\omega(L_k(\omega)) ^2\\& +(9+k ((9-2 k) k-13))(L_k(\omega))^3 \Bigr) \left(\frac{C_k}{L_k(\omega)}\right)^k\frac{  \sin \left(\frac{(1-k)\pi}{2}\right) \Gamma \left(\frac{1-k}{2}\right)}{\sqrt{\pi } \Gamma \left(3-\frac{k}{2}\right)}.\end{aligned}$$

\end{proof}
\begin{proof} [Proof of Corollary \ref{c1}.]
Let us define the probability measure $\mu_k$ on $[0,a_k]$ with density
$$h_k(t)=\frac1{|k| C_k}(a_k-t)^{\frac{1-k}2}t^{-\frac{1+k}2}(t+a_k\frac{1-k}2),$$
where
$$C_k=\frac{\pi(1-k)}{\sin(\frac{(1-k)\pi}2)},$$
and
$$a_k=\Bigl(\frac{|k| C_k\Gamma(\frac{4-k}2)}{2^{k-1}\sqrt\pi\Gamma(\frac{5-k}2)}\Bigr)^{1/(2-k)}.$$
For
$x\geq a_k$,
let $$\varphi_k(x)=x-|k|\int_{0}^{a_k}(x-t)^{k-1}\mu_k(dt)dt.$$
Differentiate yields, $$\varphi_k'(x)=1+(1-k)\int_{0}^{a_k}(x-t)^{k-1}\nu_k(dt)dt\geq 0.$$
Since,
$$x-\int_{0}^{1}(x-t)^{k-1}(1-t)^{\frac{1-k}2}t^{-\frac{1+k}2}(t+\frac{1-k}2)dt=\varphi_k(a_k)+\int_{a_k}^x\varphi_k'(t)dt.$$
and,
$$\varphi_k(a_k)=a_k\frac{3+k}8.$$
Then,
$$x-|k|\int_{0}^{a_k}(x-t)^{k-1}\mu_k(dt)dt=a_k\frac{3+k}8+\int_{a_k}^x\varphi_k'(t)dt.$$
Therefore, for $k\in(-1,1)$, $k\neq 0$
$$\frac12x^2-{\rm sgn}(k)\int_{\Bbb R_+}|x-t|^{k}\mu_k(dt)\left\{\begin{aligned}&=\sigma_k\;{\rm if}\; x\in(0,a_k)\\
&\geq \sigma_k\;{\rm if}\; x\geq a_k,\end{aligned}\right.$$
where
$$\sigma_k=-\int_{\Bbb R_+} t^{k}\mu_k(dt)=-\frac{a_k}{k C_k}\int_0^1(1-t)^{\frac{1-k}2}t^{\frac{k-1}2}(t+\frac{1-k}2)=\frac{(k-3)a_k}{8 k}.$$
Furthermore, $$E^*_k=\frac{(k-3)a_k}{8 k}+\frac12\int_{\Bbb R_+} t^2\mu_k(dt)=\frac{(k-3)a_k}{8 k}+a_k^{2}\frac{3-k}{8(4-k)}.$$
\end{proof}
For $k=0$, we get
$$h_0(t=\frac1{\pi}\sqrt{\frac{a_0-t}t}(t+\frac {a_0}2),$$
$a_0=\frac23\sqrt 6$
and
$$E_0^*=\frac34+\frac12\log 6.$$

This corresponds to the Dean-Majumdar density of positive eigenvalues for the Gaussian random matrix ensemble.

\section{Application}\label{app}

For \( k \in (-1, 2) \), define the function:
\[
V_k(x) =
\begin{cases}
\text{sign}(k) |x|^k, & k \neq 0, \\
\log|x|, & k = 0,
\end{cases}
\]
and
\[
\epsilon_k =
\begin{cases}
k, & 0 \leq k < 2, \\
1 + k, & -1 < k < 0.
\end{cases}
\]

Consider on \( \Sigma^n \) the probability density function (pdf):
\[
\mathbb{P}_n(\lambda) = \frac{1}{Z_{\Sigma, n, k}} \exp\left(-n^{1+\epsilon_k} \sum_{i=1}^n Q(\lambda_i) + n^{\epsilon_k} \sum_{1 \leq i \neq j \leq n} V_k(\lambda_i - \lambda_j)\right) d\lambda_1 \dots d\lambda_n,
\]
where \( Z_{\Sigma, n, k} \) is the normalization constant. We make the following assumptions:

\begin{enumerate}
\item[$(H_1)$] The function \( Q \) is lower semicontinuous on \( \Sigma \).
\item[$(H_2)$] For \( k \neq 0 \), \( \displaystyle \lim_{|x| \to \infty} \left(Q(x) - 4|x|^k\right) = +\infty \).
\item[$(H_3)$] For \( k = 0 \), \( \displaystyle \lim_{|x| \to \infty} \left(Q(x) - \log(x^2 + 1)\right) = +\infty \).
\end{enumerate}

For \( \Lambda = (\lambda_1, \dots, \lambda_n) \in \Sigma^n \), define the empirical measure:
\[
\mu_n^\Lambda = \frac{1}{n} \sum_{i=1}^n \delta_{\lambda_i},
\]
and let
\begin{equation}\label{e}
\mu_{\Sigma, n, k} = \mathbb{E}_n(\mu_n^\Lambda),
\end{equation}
where \( \mathbb{E}_n \) denotes the expectation with respect to \( \mathbb{P}_n \).

The goal is to study the asymptotic behavior of \( \mu_{\Sigma, n, k} \) as \( n \to \infty \). Let \( \varphi \) be a continuous function. Then:
\[
\int_\Sigma \varphi(t) \, \mu_{\Sigma, n, k}(dt) = \frac{1}{nZ_{\Sigma, n, k}}\int_{\Sigma^n} \sum_{i=1}^n \varphi(\lambda_i) \exp\left(-n^{1+\epsilon_k} \sum_{i=1}^n Q(\lambda_i) + n^{\epsilon_k} \sum_{i \neq j} V_k(\lambda_i - \lambda_j)\right) d\lambda_1 \dots d\lambda_n.
\]
By symmetry, this simplifies to:
\[
\int_\Sigma \varphi(t) \, \mu_{\Sigma, n, k}(dt) = \frac{1}{Z_{\Sigma, n, k}} \int_\Sigma \varphi(\lambda_1) h_n(\lambda_1) \, d\lambda_1,
\]
where
\[
h_n(\lambda_1) = \frac{1}{Z_{\Sigma, n, k}} \int_{\Sigma^{n-1}} \exp\left(-n^{1+\epsilon_k} \sum_{i=1}^n Q(\lambda_i) + n^{\epsilon_k} \sum_{1 \leq i \neq j \leq n} V_k(\lambda_i - \lambda_j)\right) d\lambda_2 \dots d\lambda_n.
\]

For a probability measure \( \mu \) on \( \Sigma \), define the energy:
\[
E_k(\mu) = \int_\Sigma Q(x) \, \mu(dx) - \int_\Sigma U^\mu(x) \, \mu(dx),
\]
where
\[
U^\mu(x) = \int_\Sigma V_k(x - y) \, \mu(dy).
\]
The existence of the energy ensures that \( \mu \) has compact support. The unique measure minimizing the energy is the equilibrium measure \( \mu^*_{\Sigma, k} \), which satisfies:
\[
\frac{1}{2} Q(x) - \int_\Sigma V_k(x - y) \, \mu^*_{\Sigma, k}(dy) = C_{\Sigma, k}, \quad x \in \text{supp}(\mu).
\]

Define:
\[
K_n(x) = n^{\epsilon_k} (n - 1) \sum_{i=1}^n Q(x_i) - n^{\epsilon_k} \sum_{1 \leq i \neq j \leq n} V_k(x_i - x_j),
\]
and let:
\[
h_k(s, t) = \frac{1}{2} Q(s) + \frac{1}{2} Q(t) - V_k(s - t),
\]
\[
g_k(x) = Q(x) - 4 V_k(x).
\]
Since \( V_k \) is negative definite, we have:
\[
|V_k(s - t)| \leq \left(\sqrt{|V_k(s)|} + \sqrt{|V_k(t)|}\right)^2,
\]
which implies:
\[
h_k(s, t) \geq \frac{1}{2} g_k(s) + \frac{1}{2} g_k(t).
\]
By assumptions \( H_1, H_2, H_3 \), there exists \( m_{\Sigma, k} \) such that \( g_k(s) \geq m_{\Sigma, k} \) for all \( s \in \Sigma \). Thus:
\[
h_k(s, t) \geq m_{\Sigma, k}, \quad K_n(x) \geq n^{1+\epsilon_k} (n - 1) m_{\Sigma, k}.
\]
Furthermore, \( K_n \) is lower semicontinuous on \( \Sigma^n \), and \( \lim_{\|x\| \to \infty} K_n(x) = +\infty \). Therefore, there exists \( (x_1^{(\Sigma, n, k)}, \dots, x_n^{(\Sigma, n, k)}) \in \Sigma^n \) where \( K_n \) attains its minimum.

Define the probability measure:
\[
\nu_{\Sigma, n, k} = \frac{1}{n} \sum_{i=1}^n \delta_{x_i^{(\Sigma, n, k)}},
\]
and let:
\[
\tau_{\Sigma, n, k} = \frac{1}{n^{1+\epsilon_k} (n - 1)} K_n(x_1^{(\Sigma, n, k)}, \dots, x_n^{(\Sigma, n, k)}).
\]

\begin{proposition} \
\begin{enumerate}
\item \( \displaystyle \lim_{n \to \infty} \tau_{\Sigma, n, k} = E_k(\mu^*_{\Sigma, k}) = E^*_{\Sigma, k} \).
\item The probability measure \( \nu_{\Sigma, n, k} \) converges in the tight topology to \( \mu^*_{\Sigma, k} \).
\end{enumerate}
\end{proposition}

\begin{proof}
From the definitions, we have:
\[
K_n(x) = n^{\epsilon_k} \sum_{i \neq j} h_k(x_i, x_j).
\]
For a probability measure \( \mu \),
\[
\int_{\Sigma^n} K_n(x) \, \mu(dx_1) \dots \mu(dx_n) = n^{1+\epsilon_k} (n - 1) E_k(\mu).
\]
For \( \mu = \mu^*_{\Sigma, k} \), we obtain:
\begin{equation}\label{e2}
m_{\Sigma, k} \leq \tau_{\Sigma, n, k} \leq E_k(\mu^*_{\Sigma, k}) = E^*_{\Sigma, k}.
\end{equation}
Furthermore,
\[
K_n(x) \geq n^{\epsilon_k} (n - 1) \sum_{i=1}^n g_k(x_i),
\]
which implies:
\begin{equation}\label{e1}
\int_\Sigma g_k(s) \, \nu_{\Sigma, n, k}(ds) \leq \tau_{n, k} \leq E^*_{\Sigma, k}.
\end{equation}
Thus, the sequence \( (\tau_{\Sigma, n, k}, \nu_{\Sigma, n, k}) \) is relatively compact in \( \Sigma \times \mathcal{M}^1(\Sigma) \), and there exists a subsequence \( n_j \to \infty \) such that:
\[
\lim_{j \to \infty} \tau_{\Sigma, n_j, k} = \tau_{\Sigma, k}, \quad \lim_{j \to \infty} \nu_{\Sigma, n_j, k} = \nu_{\Sigma, k}.
\]
For \( \ell > 0 \), define the truncated kernel:
\[
h^\ell_k(x, y) = \inf(h_k(x, y), \ell),
\]
and the truncated energy:
\[
E^\ell_k(\mu) = \int_{\Sigma^2} h^\ell_k(s, t) \, \mu(ds) \mu(dt).
\]
Then:
\[
E^\ell_k(\nu_{\Sigma, n, k}) \leq \tau_{\Sigma, n, k} + \frac{\ell}{n}.
\]
Taking \( j \to \infty \), we obtain:
\[
E^\ell_k(\nu_{\Sigma, k}) \leq \liminf_{j \to \infty} \tau_{\Sigma, n_j, k}.
\]
Letting \( \ell \to \infty \) and using the monotone convergence theorem, we get:
\[
E_k(\nu_{\Sigma, k}) \leq \liminf_{j \to \infty} \tau_{\Sigma, n_j, k}.
\]
From equation \eqref{e2}, we :
\[
E^*_{\Sigma, k} \leq E_k(\nu_{\Sigma, k}) \leq \liminf_{j \to \infty} \tau_{\Sigma, n_j, k} \leq \limsup_{j \to \infty} \tau_{\Sigma, n_j, k} \leq E^*_{\Sigma, k}.
\]
This proves that the sequence \( (\tau_{\Sigma, n, k}, \nu_{\Sigma, n, k}) \) converges, and:
\[
\lim_{n \to \infty} \tau_{\Sigma, n, k} = E^*_{\Sigma, k}, \quad \lim_{n \to \infty} \nu_{\Sigma, n, k} = \mu^*_{\Sigma, k}.
\]
\end{proof}

\begin{proposition}\label{pro4}
For \( k \in (-1, 2) \),
\[
\lim_{n \to \infty} \frac{1}{n^{2+\epsilon_k}} \log Z_{\Sigma, n, k} = -E^*_{\Sigma, k}.
\]
\end{proposition}

\begin{proof}
Recall that:
\[
K_n(x) = n^{\epsilon_k} (n - 1) \sum_{i=1}^n Q(x_i) - n^{\epsilon_k} \sum_{1 \leq i \neq j \leq n} V_k(x_i - x_j),
\]
and:
\[
\tau_{\Sigma, n, k} = \frac{1}{n^{1+\epsilon_k} (n - 1)} \inf_{x \in \Sigma^n} K_n(x).
\]
By the previous proposition, \( \lim_{n \to \infty} \tau_{\Sigma, n, k} = E^*_k \). Since:
\[
Z_{\Sigma, n, k} = \int_{\Sigma^n} \exp\left(-K_n(x) - n^{\epsilon_k} \sum_{i=1}^n Q(\lambda_i)\right) dx_1 \dots dx_n,
\]
we have:
\[
Z_{\Sigma, n, k} \leq \exp\left(-n^{1+\epsilon_k} (n - 1) \tau_{\Sigma, n, k}\right) \left(\int_\Sigma \exp\left(-n^{\epsilon_k} Q(x)\right) dx\right)^n.
\]
By assumptions \( H_1, H_2, H_3 \),
\[
\int_\Sigma \exp\left(-n^{\epsilon_k} Q(x)\right) dx \leq \exp\left(-(n^{\epsilon_k} - 1) \inf_{x \in \Sigma} Q(x)\right) \int_\Sigma \exp\left(-Q(x)\right) dx.
\]
Thus:
\[
\limsup_{n \to \infty} \frac{1}{n^{2+\epsilon_k}} \log \int_\Sigma \exp\left(-n^{\epsilon_k} Q(x)\right) dx = 0,
\]
and:
\[
\limsup_{n \to \infty} \frac{1}{n^{2+\epsilon_k}} \log Z_{\Sigma, n, k} \leq -E^*_{\Sigma, k}.
\]

For the lower bound, consider a probability measure \( \mu \) on \( \Sigma^n \) with density \( f(x) \), where \( f \) is continuous and supported on \( \overline{O} \), with \( f(x) > 0 \) on the open set \( O \). Then:
\[
Z_{\Sigma, n, k} = \int_O \exp\left(-K_n(x) - n^{\epsilon_k} \sum_{i=1}^n Q(x_i) - \sum_{i=1}^n \log f(x_i)\right) \prod_{i=1}^n f(x_i) dx_1 \dots dx_n.
\]
Applying Jensen's inequality:
\[
Z_{\Sigma, n, k} \geq \exp\left(-\int_O \left(K_n(x) + n^{\epsilon_k} \sum_{i=1}^n Q(x_i) + \sum_{i=1}^n \log f(x_i)\right) \prod_{i=1}^n f(x_i) dx_1 \dots dx_n\right).
\]
This implies:
\[
Z_{\Sigma, n, k} \geq \exp\left(-n^{1+\epsilon_k} (n - 1) E_k(\mu)\right) \exp\left(-\int_O \left(n^{\epsilon_k} Q(x) + \log f(x)\right) f(x) dx\right).
\]
Thus:
\[
\liminf_{n \to \infty} \frac{1}{n^{2+\epsilon_k}} \log Z_{\Sigma, n, k} \geq -E_k(\mu).
\]
If the equilibrium measure \( \mu^*_{\Sigma, k} \) has this form, the proof is complete. Otherwise, for every \( \varepsilon > 0 \), there exists a probability measure \( \mu \) of this form such that \( E_k(\mu) \leq E_k(\mu^*_{\Sigma, k}) + \varepsilon \).
\end{proof}

\begin{theorem}\label{th1}
The measure \( \mu_{\Sigma, n, k} \) defined in equation \eqref{e} converges to the equilibrium measure \( \mu^*_{\Sigma, k} \) in the tight topology:
\[
\lim_{n \to \infty} \int_\Sigma f(t) \, \mu_{\Sigma, n, k}(dt) = \int_\Sigma f(t) \, \mu^*_{\Sigma, k}(dt),
\]
for every bounded continuous function \( f \) on \( \Sigma \).
\end{theorem}

To prove the theorem, we need the following lemma:

\begin{lemma}
For \( \eta > 0 \) and \( k \in (-1, 2) \), define:
\[
A_{\eta, n, k} = \{x \in \Sigma^n \mid K_n(x) \leq (E^*_k + \eta) n^{2+\epsilon_k}\}.
\]
The set \( A_{\eta, n, k} \) is compact, and:
\[
\lim_{n \to \infty} \mathbb{P}_n(A_{\eta, n, k}) = 1.
\]
\end{lemma}

\begin{proof}
The function \( K_n \) is lower semicontinuous, so \( A_{\eta, n, k} \) is closed. From equation \eqref{e1}, we have:
\[
K_n(x) \geq n^{\epsilon_k} (n - 1) \sum_{i=1}^n g_k(x_i),
\]
where \( g_k(x) = Q(x) - 4|x|^k \). Thus:
\[
A_{\eta, n, k} \subset \{x \in \Sigma^n \mid \sum_{i=1}^n g_k(x_i) \leq \frac{n^2}{n - 1} (E^*_{\Sigma, k} + \eta)\}.
\]
By assumptions \( H_2 \) and \( H_3 \), \( A_{\eta, n, k} \) is bounded and hence compact. By definition:
\[
\mathbb{P}_n(\Sigma^n \setminus A_{\eta, n, k}) \leq \frac{1}{Z_{\Sigma, n, k}} \exp\left(-n^{2+\epsilon_k} (E^*_{\Sigma, k} + \eta)\right) \left(\int_\Sigma \exp\left(-n^k Q(x)\right) dx\right)^n.
\]
By Proposition \ref{pro4}, for all \( \varepsilon > 0 \), there exists \( N \) such that for \( n \geq N \):
\[
\frac{1}{Z_{n, k}} \leq \exp\left(n^{2+\epsilon_k} (E^*_{\Sigma, k} + \varepsilon)\right).
\]
Choosing \( \varepsilon < \eta \), the result follows.
\end{proof}

\begin{proof}[Proof of Theorem \ref{th1}]
Let \( f \) be a bounded continuous function on \( \Sigma \), and define:
\[
F_n(x) = \frac{1}{n} \sum_{i=1}^n f(x_i).
\]
Fix \( \eta > 0 \) and \( k \in (-1, 2) \). The set \( A_{\eta, n, k} \) is compact, so \( F_n \) attains its supremum on \( A_{\eta, n, k} \) at some point \( x^{(\eta, n, k)} = (x_1^{(\eta, n, k)}, \dots, x_n^{(\eta, n, k)}) \). Thus:
\[
\int_\Sigma f(t) \, \mu_{\Sigma, n, k}(dt) \leq F_n(x^{(\eta, n, k)}) + \|f\|_\infty (1 - \mathbb{P}_n(A_{\eta, n, k})).
\]
Define the probability measure:
\[
\nu_{n, k}^{(\eta)} = \frac{1}{n} \sum_{i=1}^n \delta_{x_i^{(\eta, n, k)}}.
\]
The inequality above becomes:
\begin{equation}\label{e3}
\int_\Sigma f(x) \, \mu_{\Sigma, n, k}(dx) \leq \int_\Sigma f(t) \, \nu_{n, k}^{(\eta)}(dt) + \|f\|_\infty (1 - \mathbb{P}_n(A_{\eta, n, k})).
\end{equation}
The truncated energy \( E^\ell_k \) of \( \nu_{n, k}^{(\eta)} \) satisfies:
\[
E^\ell_k(\nu_{n, k}^{(\eta)}) \leq E^*_k + \eta + \frac{\ell}{n}.
\]
From equation \eqref{e1}, we have:
\[
\int_\Sigma g_k(t) \, \nu_{n, k}^{(\eta)}(dt) \leq \frac{n}{n - 1} (E^*_{\Sigma, k} + \eta).
\]
This implies that the sequence \( \nu_{n, k}^{(\eta)} \) is relatively compact in the tight topology. There exists a subsequence \( n_j \to \infty \) such that:
\[
\lim_{j \to \infty} \nu_{n_j, k}^{(\eta)} = \nu_k^{(\eta)}.
\]
We may also assume:
\[
\lim_{j \to \infty} \int_\Sigma f(t) \, \mu_{\Sigma, n_j, k}(dt) = \limsup_{n \to \infty} \int_\Sigma f(t) \, \mu_{\Sigma, n, k}(dt).
\]
The limit measure \( \nu_k^{(\eta)} \) satisfies:
\[
E^\ell_k(\nu_k^{(\eta)}) \leq E^*_{\Sigma, k} + \eta.
\]
Letting \( \ell \to \infty \), we obtain:
\[
E_k(\nu_k^{(\eta)}) \leq E^*_{\Sigma, k} + \eta.
\]
Moreover, from equation \eqref{e3}:
\[
\limsup_{n \to \infty} \int_\Sigma f(t) \, \mu_{\Sigma, n, k}(dt) \leq \int_\Sigma f(t) \, \nu_k^{(\eta)}(dt).
\]
The inequality \( E_k(\nu_k^{(\eta)}) \leq E^*_{\Sigma, k} + \eta \) implies that \( \nu_k^{(\eta)} \) converges to \( \mu^*_{\Sigma, k} \) as \( \eta \to 0 \). Therefore:
\[
\limsup_{n \to \infty} \int_\Sigma f(t) \, \mu_{\Sigma, n, k}(dt) \leq \int_\Sigma f(t) \, \mu^*_{\Sigma, k}(dt).
\]
Applying this inequality to \( -f \) instead of \( f \), we obtain:
\[
\liminf_{n \to \infty} \int_\Sigma f(t) \, \mu_{\Sigma, n, k}(dt) \geq \int_\Sigma f(t) \, \mu^*_{\Sigma, k}(dt).
\]
Thus:
\[
\lim_{n \to \infty} \int_\Sigma f(t) \, \mu_{\Sigma, n, k}(dt) = \int_\Sigma f(t) \, \mu^*_{\Sigma, k}(dt).
\]
\end{proof}

\begin{corollary}
\begin{enumerate}
\item For \( Q(x) = |x|^\beta \) and \( \Sigma = \mathbb{R} \), the measure \( \mu_{\Sigma, k, n} \) converges to the probability measure \( \mu^*_{k, \beta} \) of Proposition \ref{pa}.
\item For \( Q(x) = x^2 \) and \( \Sigma = [\omega, +\infty) \), the measure \( \mu_{\omega, n, k} \) converges to the equilibrium measure \( \mu^*_{k, \omega} = \nu_{k, \beta} \) of Theorem \ref{tii}.
\end{enumerate}
\end{corollary}
\section{Density of the maximal particle}
In this section, we focus on the position $x_{\text{max}}$ of the rightmost particle on the infinite
line. Clearly  $x_{\text{max}}$ is a random variable which fluctuates from one realisation to another. To derive the distribution of  $x_{\text{max}}$, it is convenient to consider the cumulative distribution
\[ D(\omega, n)=\text {Prob}( x_{\text{max}}\leq\omega)={\Bbb P}_n[x_1\leq\omega,...,x_n\leq \omega].\]
By the symmetry of the density of the particles
\[\text {Prob}( x_{\text{max}}\leq\omega)=\text {Prob}( x_{\text{min}}\geq-\omega)\]
Using the Boltzmann distribution of the particles \({\Bbb P}_n(x_1, x_2, ..., x_n )\), one can express $D(\omega, n)$
as the ratio of two partition functions:
\[D(-\omega, n)=\frac{Z_{k,n}(-\omega)}{Z_{k, n}(-\infty)},\]
where \[Z_{k,n}(-\omega)=\int_{[-\omega,\infty)^n}{\Bbb P}_n(x_1,...,x_n)dx_1...dx_n.\]
Using the result of Sections~[\ref{den}, \ref{app}]), we get
For $k\in(-1,1)$, \\
\(\bullet\) If $\omega\geq \omega_k$,
\[\lim_{n\to\infty}-\frac1{n^{2+\epsilon_k}}\log D(\omega, n)=1.\]
\(\bullet \) If $\omega\leq \omega_k$,
\[\lim_{n\to\infty}-\frac1{n^{2+\epsilon_k}}\log D(\omega, n)=E^*_{k, \omega}-E^*_k.\]
In particular the probability \(p_{k}\) that all particles are positive is the same as the probability that all particles are negative and is given by
\[
p_k=\lim_{n\to\infty}-\frac1{n^{2+\epsilon_k}}\log D(0, n)=E^*_{k,0}-E^*_k = \frac{\left(1 - (3-k)^{-\frac{k}{2-k}}\right)(2-k)}{k(4-k)} \omega_k^2.
\]
For \(k = 0\),
\[
 p_{0} = \frac{1}{2}\log 3.
\]

\section{Generalized Marchenko-Pastur Distribution}

Consider the potential \( Q(x) = x \) defined on the interval \([0, +\infty)\). For this potential, we define a probability density on \(\mathbb{R}_+^n\) as follows:
\[
\mathbb{P}_{n,k}(dx) = \frac{1}{Z_{n,k}} \prod_{i=1}^n e^{-n^{1+\epsilon_k} x_i} \prod_{1 \leq i < j \leq n} e^{n^{\epsilon_k} V_k(x_i - x_j)} dx_1 \dots dx_n,
\]
where \( k \in (-1, 1) \), \(\epsilon_k\) is a parameter, and \( Z_{n,k} \) is a normalization constant.

The empirical density of the particles is given by:
\[
\nu_n = \frac{1}{n} \sum_{i=1}^n \delta_{x_i},
\]
and the statistical density of the particles is defined as:
\[
\mu_{n,k} = \mathbb{E}_{n,k}(\nu_n),
\]
where \(\mathbb{E}_{n,k}\) denotes the expectation with respect to the probability density \(\mathbb{P}_{n,k}\).

For \( k \in (-1, 1) \) with \( k \neq 0 \), we define the probability measure \(\mu_k\) with density:
\[
h_k(t) = \frac{\sin\left(\frac{(1-k)\pi}{2}\right)}{2|k|\pi} (a_k - t)^{\frac{1-k}{2}} t^{-\frac{1+k}{2}} \chi_{[0, a_k]}(t),
\]
where
\[
a_k = \left(\frac{|k| \Gamma\left(\frac{1+k}{2}\right) \Gamma\left(\frac{2-k}{2}\right)}{2^{k-2} \sqrt{\pi}}\right)^{\frac{1}{1-k}}.
\]
For \( k = 0 \), the density is:
\[
h_0(t) = \frac{1}{2\pi} \sqrt{\frac{4 - t}{t}} \chi_{[0, 4]}(t).
\]
\begin{proposition}\label{th0}
Let \( k \in (-1, 1) \). Consider the energy functional
\[
E_k(\mu) = \int_{\Bbb R_+} x \mu(dx) - \int_{\Bbb R_+^2} V_k(x - y) \mu(dx) \mu(dy).
\]
Then
\[
\inf_{\mu \in {\cal M}^1(\Bbb R)} E_k(\mu) = E_k(\mu_{k}) = E^*_{k}=-\frac{1}{2} \frac{(1 - k)^2}{k(2 - k)} a_k\qquad, k\neq 0.
\]
\end{proposition}

\begin{proof}
To demonstrate that the probability measure \(\mu_k\) minimizes the energy, it suffices to show that $\mu_k$ satisfies the conditions of Proposition \ref{p3}:
\[
\frac{1}{2}x - \text{sgn}(k) \int_0^\infty |x - t|^k \mu_k(dt) \begin{cases}
= \sigma_k & \text{if } x \in (0, a_k), \\
\geq \sigma_k & \text{if } x \geq a_k.
\end{cases}
\]
From \cite{1}, Corollary 2.2, for integers \( m, n, r \) with \( m + r + n + 1 > 0 \), \(\nu \in (0, 1)\), and \( x \in (-1, 1) \),
\[
\int_{-1}^1 \frac{\text{sign}(x - t)}{|x - t|^\nu} Q_m^{\nu/2 + r, \nu/2 + n}(t) dt = v_{n,m,r,\nu} P_{m + r + n + 1}^{\nu/2 - r - 1, \nu/2 - k - 1}(t),
\]
where
\[
v_{n,m,r,\nu} = \frac{\pi (-1)^r 2^{r + n + 1} \Gamma(m + \nu)}{m! \Gamma(\nu) \sin\left(\frac{\nu \pi}{2}\right)}.
\]
For \( y \in (-1, 1) \),
\[
2 \int_0^1 \frac{\text{sign}(y + 1 - 2t)}{|y + 1 - 2t|^\nu} Q_m^{\nu/2 + r, \nu/2 + k}(2t - 1) dt = v_{n,m,r,\nu} P_{m + r + n + 1}^{\nu/2 - r - 1, \nu/2 - k - 1}(y).
\]
For \( x = \frac{y + 1}{2} \in (0, 1) \),
\[
2^{1 - \nu} \int_0^1 \frac{\text{sign}(x - t)}{|x - t|^\nu} Q_m^{\nu/2 + r, \nu/2 + n}(2t - 1) dt = v_{n,m,r,\nu} P_{m + r + n + 1}^{\nu/2 - r - 1, \nu/2 - n - 1}(2x - 1).
\]
For \( n = -1 \), \( r = 0 \), and \( \nu = 1 - k \),
\[
2^k \int_0^1 \text{sign}(x - t) |x - t|^{k - 1} Q_m^{(1 - k)/2, -(1 + k)/2}(2t - 1) dt = \frac{\pi \Gamma(m + 1 - k)}{m! \Gamma(1 - k) \sin\left(\frac{(1 - k)\pi}{2}\right)} P_m^{-(1 + k)/2, (1 - k)/2}(2x - 1).
\]
For \( m = 0 \),
\[
\int_0^1 \text{sign}(x - t) |x - t|^{k - 1} (1 - t)^{\frac{1 - k}{2}} t^{-\frac{1 + k}{2}} dt = \frac{\pi}{\sin\left(\frac{(1 - k)\pi}{2}\right)}.
\]
Thus, for \( x \in (0, 1) \) and \( k \in (-1, 1) \), \( k \neq 0 \),
\[
\frac{1}{2}x - \text{sgn}(k) \int_0^1 |x - t|^k \nu_k(dt) = c_k.
\]
For \( k = 0 \),
\[
\frac{1}{2}x - \int_0^1 \log|x - t| \nu_0(dt) = c_k.
\]
Here, \(\nu_k\) is the probability measure with density \( h_k \), where
\[
h_k(t) = \frac{\sin\left(\frac{(1 - k)\pi}{2}\right)}{2|k|\pi} (1 - t)^{\frac{1 - k}{2}} t^{-\frac{1 + k}{2}} \chi_{[0, 1]}(t),
\]
and
\[
h_0(t) = \frac{1}{2\pi} \sqrt{\frac{4 - t}{t}} \chi_{[0, 4]}(t).
\]
For \( x \in (0, a_k) \), using the change of variables \( x = a_k y \) and \( t = a_k u \), we obtain:
\[
\frac{1}{2}x - \text{sgn}(k) \int_0^{a_k} |x - t|^k \mu_k(dt) = a_k \left(\frac{1}{2} y - \text{sgn}(k) \int_0^1 |x - u|^k \nu_k(du)\right) = a_k c_k = \sigma_k,
\]
where \(\mu_k\) is the probability measure with density:
\[
h_k(t) = \frac{\sin\left(\frac{(1 - k)\pi}{2}\right)}{2|k|\pi} (a_k - t)^{\frac{1 - k}{2}} t^{-\frac{1 + k}{2}} \chi_{[0, a_k]}(t),
\]
and
\[
a_k = \left(\frac{|k| \Gamma\left(\frac{1 + k}{2}\right) \Gamma\left(\frac{2 - k}{2}\right)}{2^{k - 2} \sqrt{\pi}}\right)^{\frac{1}{1 - k}}.
\]

For \( x \geq a_k \), let
\[
\psi(x) = \frac{1}{2}x - \text{sgn}(k) \int_0^{a_k} |x - t|^k \mu_k(dt).
\]
Then,
\[
\psi''(x) = -k(k - 1) \text{sgn}(k) \int_0^{a_k} |x - t|^{k - 2} \mu_k(dt) \geq 0.
\]
Thus,
\[
\psi'(x) \geq \psi'(a_k) = \frac{1}{2} - k \text{sgn}(k) \int_0^1 |x - t|^{k - 1} \nu_k(dt) = 0.
\]
Therefore, \(\psi\) is increasing and \(\psi(x) \geq \psi(a_k) = \sigma_k\).

Since
\[
\sigma_k = -\text{sgn}(k) \int_0^{a_k} t^k \mu_k(dt) = -\frac{1 - k}{4k} a_k,
\]
and
\[
\frac{1}{2} \int_0^1 t \mu_k(dt) = \frac{1 - k}{4(2 - k)} a_k,
\]
the energy is given by:
\[
E^*_k = \sigma_k + \frac{1}{2} \int_0^1 t \mu_k(dt) = -\frac{1}{2} \frac{(1 - k)^2}{k(2 - k)} a_k.
\]
For \( k = 0 \),
\[
\sigma_0 = -\int_0^4 \log t \mu_0(dt) = -1,
\]
and
\[
\int_0^4 t \mu_0(dt) = 1.
\]
Thus,
\[
E_0 = \sigma_0 + \frac{1}{2} \int_0^4 t \mu_0(dt) = -\frac{1}{2}.
\]
\end{proof}
\begin{proposition}
For \( k \in (-1, 1) \), the probability measure \(\mu_{n,k}\) converges for the tight topology to the probability measure \(\mu_k\). For every bounded continuous function \( f \) on \(\mathbb{R}_+\),
\[
\lim_{n \to \infty} \int_{\mathbb{R}_+} f\left(\frac{t}{n^{1+\epsilon_k}}\right) \mu_{n,k}(dt) = \int_{\mathbb{R}_+} f(t) \mu_k(dt).
\]
\end{proposition}

\end{document}